\documentclass[11pt]{article}
\usepackage{graphicx}
\usepackage{bm}
\usepackage[ruled,noline]{algorithm2e}

\usepackage{latexsym}
\usepackage{xspace}
\usepackage{amsmath}
\usepackage{amsthm}
\usepackage{amssymb}
\usepackage{fancybox}
\usepackage{float}
\floatstyle{boxed}
\restylefloat{figure}

\usepackage[letterpaper,hmargin=1.0in,tmargin=1.0in,bmargin=1.0in]{geometry}
\usepackage{url}
\usepackage{soul,xcolor}
\usepackage{multirow}

\newtheorem{thm}{Theorem}  
\newtheorem{lem}{Lemma}
\newtheorem{lemma}{Lemma}
\newtheorem{cor}[lemma]{Corollary}
\newtheorem{define}[lemma]{Definition}

\newcommand\numberthis{\addtocounter{equation}{1}\tag{\theequation}}

\newcommand{\argmax}{\text{argmax}}
\DeclareMathOperator{\CR}{CR}

\newenvironment{enumerate*}%
  {\vspace{-2ex} \begin{enumerate} %
     \setlength{\itemsep}{-1ex} \setlength{\parsep}{0pt}}%
  {\end{enumerate}}

\newenvironment{itemize*}%
  {\vspace{-2ex} \begin{itemize} %
     \setlength{\itemsep}{-1ex} \setlength{\parsep}{0pt}}%
  {\end{itemize}}

\newenvironment{description*}%
  {\vspace{-2ex} \begin{description} %
     \setlength{\itemsep}{-1ex} \setlength{\parsep}{0pt}}%
  {\end{description}}
 
\newcommand{\hide}[1]{}
  
\newcommand{\rjc}[1]{\textcolor{blue}{#1}}
\newcommand{\skr}[1]{\textcolor{red}{#1}}

\title{Applications of $\alpha$-strongly regular distributions to Bayesian auctions}
\author{Richard Cole\thanks{The work of Richard Cole work was supported in
 part by NSF Grants CCF-1217989 and CCF-1527568.} \and Shravas Rao\thanks{The work of Shravas Rao was supported in part by the National Science Foundation Graduate Research Fellowship Program under Grant No. DGE-1342536.}}
\date{\today}

\begin{document}

\maketitle
\thispagestyle{plain}
\pagestyle{plain}

\begin{abstract}
Two classes of distributions that are widely used in the analysis of Bayesian auctions
are the Monotone Hazard Rate (MHR) and Regular distributions.
They can both be characterized in terms of the rate of change of the associated virtual value functions:
for MHR distributions the condition is that for values $v < v'$, $\phi(v') - \phi(v) \ge v' - v$,
and for regular distributions, $\phi(v') - \phi(v) \ge 0$.
Cole and Roughgarden introduced the interpolating class of $\alpha$-Strongly Regular distributions
($\alpha$-SR distributions for short),
for which $\phi(v') - \phi(v) \ge \alpha(v' - v)$, for $0 \le \alpha \le 1$.

In this paper, we investigate five distinct auction settings for which
good expected revenue bounds are known when the bidders' valuations are given by MHR
distributions. In every case, we show that these bounds degrade gracefully 
when extended to $\alpha$-SR distributions.
For four of these settings, the auction mechanism requires knowledge of these distribution(s)
(in the other setting, the distributions are needed only to ensure good bounds on the
expected revenue).
In these cases we also investigate
what happens when the distributions are known only approximately via samples, specifically
how to modify the mechanisms so that they remain effective and how
the expected revenue depends on the number of samples.
\end{abstract}

\section{Introduction}
\label{sec:intro}
Much of the recent computer science research on revenue-maximizing auctions uses
Bayesian analysis to measure auction performance (see~\cite{HartBook} for an overview),
although there is also a considerable body of work on worst-case revenue maximization (see~\cite{HK07}).
Typically the analyses seek to compare the revenue for the given mechanism to a measure of the
optimal revenue, expressing this as an approximation factor.

In Bayesian analyses the bidders valuations are assumed to be drawn from one or more distributions,
either one common distribution for all the bidders, or separate distributions for distinct
groups of bidders, possibly with each bidder being in a distinct group.
Almost all previous Bayesian analyses have been for one of three settings:
all distributions, regular distributions, and Monotone Hazard Rate (MHR) distributions,
with MHR being the more restrictive.
For example, Myerson's analysis~\cite{M} of the expected revenue of the optimal auction
for the sale of a single item is most natural when the buyer values are given
by regular value distributions (different buyers may have values drawn from
distinct distributions).
Many other results, including those we will consider in this paper,
are currently known only for MHR distributions, and for the most part do not
extend to regular distributions.

Recently, Cole and Roughgarden~\cite{CR1} introduced the notion of $\alpha$-Strongly Regular distributions,
$\alpha$-SR distributions for short;
these interpolate between MHR and Regular distributions.
They gave two examples of settings for which results previously shown for MHR distributions
extended smoothly to $\alpha$-SR distributions.
However, the main focus of their work was to investigate what happens in auctions,
and in Myerson's auction in particular, when distributions are known only approximately,
rather than exactly, and how to analyze the resulting expected revenue as a function of the
number of samples.

In independent work, Schweizer and Szech identified the same notion, which they term $\lambda$-regularity, where $\lambda$ corresponds to $1-\alpha$~\cite{SS15}.  They proved this is equivalent to $\rho$-concavity, an earlier, albeit less intuitive characterization of this class of distributions.

In this paper we carry out a more thorough investigation of $\alpha$-SR distributions,
and specifically to what extent known results for
MHR distributions extend to $\alpha$-SR distributions.
We consider five auction settings, listed in Table~\ref{tbl:results}.
For each problem, we show that the prior result extends smoothly.
In addition, for four of these problems, the auction uses knowledge of the distribution
in its decision making.
For these two settings, we propose variants of the auctions which allow efficiency
in terms of revenue to be maintained, and we also determine how the expected revenue
varies as a function of the number of samples.

The technical challenges in this work were two-fold.
First, we had to extend a variety of results concerning properties of MHR distributions
to $\alpha$-SR distributions.
While some of these results are straightforward extensions of analogous results
for MHR distributions, in other cases new proofs were needed, as the previous
arguments depended on convexity properties that need not hold outside the MHR domain.
For the most part, once these new results were obtained, analyzing the auction revenue
was simply a matter of replacing an MHR bound with the corresponding $\alpha$-SR bound,
as illustrated 
in Section~\ref{approx-algs}.

Second, in working with samples we had to adjust some of the mechanisms to take account of the fact that
they were using an approximation of the actual distributions.
For example, for the result in Theorem~\ref{thm-BGGm-w-sampling}, we take the apparently
optimal solution based on the approximate distributions,
and
adjust it in a non-uniform manner;
the resulting solution achieves an approximation factor similar to what is obtained given exact distributions.

Schweizer and Szech in~\cite{SS15} also prove results concerning $\alpha$-SR distributions.  
For the most part, these results have a more structural flavor than the results in this paper. 
However, the main bound in their Lemma 6 is the same as Lemma~\ref{minmax} in this paper, and the bound in their Proposition 4 is the same as Lemma~\ref{welfare} in this paper (the proofs appear to be quite different, however).  While they do not delve into the specific bounds achieved in applications, they identify an essentially disjoint set of applications to which the $\alpha$-SR generalization of MHR bounds can be applied, including those found in~\cite{AGM09, ADMW13, CDHS14, DP11, HMS08}.  

Finally, Schweizer and Szech develop a notion of $\lambda^*$-regularity, generalizing $\lambda$-regularity. This is an alternative to hyper-regularity, a notion proposed by Kleinberg and Yuan~\cite{KY13}, which had been used to obtain revenue to welfare bounds in single parameter mechanisms.

In sum, this work strongly suggests that results that hold w.r.t.\ MHR distributions will often degrade
gracefully when extended to $\alpha$-SR distributions.  The one result we did not succeed in extending was Theorem 3.14 in~\cite{DRY}.
It would be interesting to know if there are problems for which there is no graceful degradation.  This work also suggests that the optimal mechanism given full knowledge of the distributions may need non-trivial modifications to achieve good performance when faced with sample-based empirical distributions.

\begin{table}[tb]
\label{tbl:results}
\centering
\begin{tabular}{l  c  c  c}   \hline
Mechanism        &  MHR                          & $\alpha$-SR                  &  with samples \\ \hline\hline
VCG for  & \multirow{2}{*}{3 \cite{HR}  }
            & \multirow{2}{*}{$\frac{2 + \alpha} {\alpha}$ (Thm.~\ref{thm:VCG-w-dulpl})} & \multirow{2}{*}{n/a} \\
downward closed &&& \\
revenue approx. \\ \hline
VCG-L mechanism  &   \multirow{3}{*}{$\frac{1}{e}$ \cite{DRY}}       & \multirow{3}{*}{ ${\alpha^{1/(1 - \alpha)}}$ (Thm.~\ref{thm:VCGL})}   & \multirow{3}{*}{result below} \\
revenue vs.\ VCG & & & \\
welfare &&& \\ 
$k$-bidders &&  \multicolumn{2}{c}{~~~~~~~~~~~~~~$\alpha^{1/(1-\alpha)}\cdot\frac{1-\xi(1+\gamma)^2(1-k\delta)}{(1+\gamma)^4}$ (Thm.~\ref{thm:vcgl-samp})} \\ \hline
%
%
Downward closed, & \multirow{3}{*}{ $4(1 + e)$ \cite{CMM} }
&  \multirow{3}{*}{~~~$\frac{4}{\alpha}+2\frac{\alpha+1}{\alpha^{(2-\alpha)/(1-\alpha)}}$ (Thm.~\ref{thm:two-mech})~~~}
   & \multirow{2}{*}{result below} \\
 known budgets  &&& \\
 social welfare approx. &&&\\
&&\multicolumn{2}{c}{~~~~~~~~~~$\frac{4}{\alpha}+2\frac{\alpha+1}{\alpha^{(2-\alpha)/(1-\alpha)}(1-\epsilon)(1-n\delta)}$ (Thm.~\ref{thm:two-mech-samp})} \\
 \hline
Downward closed,  & \multirow{4}{*}{ $3(1 + e)$~\cite{CMM}}
         &\multirow{4}{*}{ $3 \left( 1 + \frac {1} {\alpha^{1/(1 - \alpha)}} \right)$  (Thm.~\ref{lottery})}
         & \multirow{4}{*}{ result below } \\
private budgets, &&& \\
single parameter &&& \\
revenue approx. &&& \\
$k$-bidders &  \multicolumn{3}{c}
 { ~~~~~~~~~~~$\frac{3} {1 - k \delta} \left( 1 + \frac {1} {\alpha^{1/(1 - \alpha)}(1-\max\{\sqrt{8\gamma/\alpha}, 4\gamma+\xi\})} \right)$
 (Thm.~\ref{thm:lottery-samp})} \\  \hline
Public budget, &\multirow{3}{*}{ $192e^2$ \cite{BGGM} }
             & \multirow{3}{*}{ $\frac{192} {\alpha} \left( \frac {2 - \alpha} {\alpha} \right)^{1/(1 - \alpha)}$  (Thm.~\ref{lpmech})}
             & \multirow{3}{*}{ see Thm.~\ref{thm-BGGm-w-sampling} } \\
universally IC &&& \\
revenue approx. \\ \hline\hline
&&&\\
\end{tabular}
\caption{Results for the mechanisms we analyze.}
\end{table}

Our results are shown in Table 1.
Column 2 gives known results for MHR distributions
expressed as an approximation factor; Column 3 gives the corresponding results for $\alpha$-SR distributions, and Column 4, where applicable, the results under sampling of the distributions.
$\delta$, $\xi$, and $\gamma$ are parameters used to specify the number $m$ of samples and which need to satisfy $\gamma \xi m \geq 4$, $(1+\gamma)^2 \leq 3/2$, and $m \geq \frac{6(1+\gamma)}{\gamma^2\xi} \max\{\frac{\ln{3}}{\gamma}, \ln\frac{3}{\delta}\}$.  Reasonable choices are $\xi = \delta$, $\delta = \gamma / k$, and $\gamma \leq 1/5$ as small as needed to give the desired approximation factor (except for Theorem~\ref{thm:two-mech-samp}, where $\delta = \gamma/n$).  All the sampling results assume there are $k$ classes of bidders each with their own distribution, and $n$ bidders in total.  Note that when $\alpha$ tends to $1$, the limit values for all the bounds in column 3, 
 are the prior known bounds for MHR distributions.

Our goal with this work is two fold.  First, we aim to show that results for MHR distributions can often be extended to $\alpha$-SR distributions.  Second, by providing a tool-kit of results about $\alpha$-SR distributions we hope to encourage other authors to attempt to extend their MHR results to $\alpha$-SR distributions.

In Section~\ref{sec:prelims} we review some standard definitions and results.
In Section~\ref{approx-algs} we prove Theorems~\ref{thm:VCG-w-dulpl}--\ref{lpmech} (the results in column 3 in Table 1).
%
In Section~\ref{sec:sampling} we explore what happens when the distributions are known approximately via samples, proving Theorems~\ref{thm:vcgl-samp}--\ref{thm-BGGm-w-sampling} (the results in column 4).

\section{Preliminaries}
\label{sec:prelims}


Recall that for a distribution $F$, the virtual valuation $\phi(v)$ is given by
\[\phi(v) = v - \frac{1-F(v)}{f(v)},\]
where $f$ is the derivative of $F$.  Sometimes, we might define $F$ on a discrete set $\{1, \ldots, L\}$,
for some $L$, in which case we define the virtual valuation as
\[\phi(v) = v-\frac{1-F(v)}{F(v)-F(v-1)},\]
 where $F(0) = 0$.  Unless otherwise stated, we will assume that $F$ is a continuous distribution.
It is often useful to use the hazard rate, $h(v) = f(v)/(1-F(v))$ (or $h(v) = (F(v)-F(v-1))/(1-F(v))$ in the case of a discrete distribution);
then $\phi(v) = v - 1/h(v)$.
Note that $f$, $F$, and $h$ are always non-negative.

Given a value $v$, it can be useful to refer to the quantile, $q(v) = 1-F(v)$.  Additionally, we let $v(q)$ be the value at quantile $q$.

Also recall that the monopoly price  is the least price $r$ such that $\phi(r) \geq 0$.

The following definition of $\alpha$-SR distributions was introduced in~\cite{CR1}.

\begin{define}\label{alpha}
A distribution  $F$ is $\alpha$-SR if for all $x < y$, \[\phi(y) - \phi(x) \geq \alpha(y-x).\]
\end{define}

Note that monotone-hazard (MHR) rate distributions are $1$-SR, and regular distributions are $0$-SR.  If $F$ is a continuous distribution, then Definition~\ref{alpha} is equivalent to stating that $\frac{d\phi}{dv} \geq \alpha$.

The following worst-case $\alpha$-SR distributions, first given in~\cite{CR1}, will be used to show that several of our results are tight:
\begin{equation*}
F^{\alpha}(v) = 1-\left(1+\frac{1-\alpha}{\alpha}v\right)^{-\frac{1}{1-\alpha}},~~~
f^{\alpha}(v) = \frac{1}{\alpha} \left( 1 + \frac{1-\alpha}{\alpha} v \right) ^ {-\frac{2-\alpha}{1-\alpha}}.
\end{equation*} 

These distributions have power-law tails with parameter $c = 2+ \frac{\alpha}{1-\alpha}$, i.e. $f^{\alpha}(v) = \theta(v^{-c})$ for large $v$.

\section{Approximation Algorithms for $\alpha$-SR Distributions}
\label{approx-algs}

The versions of all of Theorems~\ref{thm:VCG-w-dulpl}--~\ref{lpmech} 
for MHR distributions rely on various quantitative properties of MHR distributions.  The new results depend on generalizing these properties to $\alpha$-SR distributions; some of these extensions are quite non-trivial.

\subsection{Revenue of VCG with Duplicates}

Theorem~\ref{thm:VCG-w-dulpl} bounds the expected revenue of Vickrey-Clarke-Groves (VCG) with duplicates as described in~\cite{HR}.    Recall that the VCG mechanism chooses the feasible set of bidders with the maximum total value to be the winners, and charges each bidder appropriately.  With duplicates, VCG is run on the set of bidders, along with a single additional copy of each bidder, so that each bidder and its copy have independent and identical distributions on their valuations, are interchangable, and cannot both be part of the winning set of bidders. In Theorem~\ref{thm:VCG-w-dulpl}, as $\alpha$ tends to $1$, our bound on the approximation factor tends to $3$,
the tight bound previously achieved for
MHR distributions in~\cite{HR}.

\begin{thm}
\label{thm:VCG-w-dulpl}
Let $0 < \alpha < 1$.  For every downward-closed environment with valuations drawn independently from distributions that are $\alpha$-SR, the expected revenue of VCG with duplicates
is a $ \left(\frac{2+\alpha}{\alpha}\right)$-approximation to the expected revenue of the optimal mechanism without duplicates.
\end{thm}
\begin{proof}
Lemma~\ref{conditioned} below replaces Lemma 4.1 in the proof of Theorem 4.2 in~\cite{HR}.  The rest of the proof is unchanged.
\end{proof}

\begin{lem}\label{conditioned}
Let $0 < \alpha < 1$ and let $F$ be an $\alpha$-SR distribution, and $\phi$ be its virtual valuation function.
Then, for all $t$,
\begin{align*}
&\mathbb{E}_{v_1, v_2 \sim F}[\max\{v_1, v_2\} | \max\{v_1, v_2\} \geq t] \leq \\
&~~~~~~~~~\left(\frac{2+\alpha}{\alpha}\right)\mathbb{E}_{v_1, v_2 \sim F}[\max(\phi(v_1), \phi(v_2)) | \max\{v_1, v_2\} \geq t].
\end{align*}
\end{lem}

In Lemma~\ref{minmax} below we prove Lemma~\ref{conditioned} for the case $t = 0$, which it turns out is when the bound is tightest, as we later show in concluding the proof of Lemma~\ref{conditioned}.

To prove Lemma~\ref{minmax} we will use the following structural properties of $\alpha$-SR distributions $F$ and their density functions $f$.

\begin{enumerate}
\item A lower bound on $f(q)$ (Lemma~\ref{densitybound}): for $q \leq q_0 \leq 1$, $f(q) \geq f(q_0)\left(\frac{q}{q_0}\right)^{2-\alpha}.$

\item The \emph{single crossing property} (Lemma~\ref{cdfbound}):  if for some $v_0$, $F(v_0) > F^{\alpha}(v_0)$, then $F(v) \geq F^{\alpha}(v)$ for all $v \geq v_0$, where $F^{\alpha}$ is a tight distribution: $f(q) = f(1)\cdot q^{2-\alpha}.$

\item Lemma~\ref{square}: $\frac{\alpha}{1+\alpha} 
\int_0^{\infty} (1-F(v))dv \leq \int_0^{\infty} (1-F(v))^2dv.$
\end{enumerate}


\begin{lem}~\label{densitybound}
Let $F$ be an $\alpha$-SR distribution, and let $f$ be the density function for $v$.  Let $q_0 \in [0, 1]$.  Then for $q \leq q_0$, \[f(q) \geq f(q_0)\left(\frac{q}{q_0}\right)^{2-\alpha}.\]
\end{lem}
\begin{proof}
Recall that $\phi(v) = v-q(v)/f(v)$ and $\frac{dv}{dq} = -1/f(v)$.  Hence $\frac{d\phi}{dq} = \frac{-2}{f(v)}+\frac{q(v)}{f(v)^2}\frac{df}{dq}$.  The condition $\frac{d\phi}{dv} \geq \alpha$ yields $\frac{d\phi}{dq} = \frac{d\phi}{dv}\frac{dv}{dq} \leq \frac{-\alpha}{f(v)}$.  Thus $\frac{q(v)}{f(v)^2}\frac{df}{dq} \leq \frac{2-\alpha}{f(v)}$ or $\frac{d}{dq}\ln{f} \leq \frac{1}{q}(2-\alpha)$.  For $q \leq q_0$, this yields $f(q) \geq f(q_0)\left(\frac{q}{q_0}\right)^{2-\alpha}$ as desired.
\end{proof}


\begin{lem}\label{cdfbound}
Let $F$ be an $\alpha$-SR distribution, and let $F^{\alpha}$ be an $\alpha$-SR distribution such that $f(q) = f(q_0)\left(\frac{q}{q_0}\right)^{2-\alpha}$ for all $q_0 \in [0, 1]$ and $q \leq q_0$. If for some $v_0$, $F(v_0) > F^{\alpha}(v_0)$, then $F(v) \geq F^{\alpha}(v)$ for all $v \geq v_0$.
\end{lem}
\begin{proof}
Assume for a contradiction that the statement of the lemma does not hold.  In particular, assume that there are $v_2$ and $v_4$ with $v_4 > v_2$, $F(v_2) > F^{\alpha}(v_2)$, but $F(v_4) < F^{\alpha}(v_4)$.  Then there must exist $v_3$, with $v_2 < v_3 < v_4,$ such that the function $F(v)-F^{\alpha}(v)$ crosses the $x$-axis from above at $v = v_3$.  It follows that $f(v_3)- f^{\alpha}(v_3) < 0,$ where $f$ and $f^{\alpha}$ are the density functions, or derivatives, of $F$ and $F^{\alpha}$ respectively.

Suppose that the function $F(v)-F^{\alpha}(v)$ crosses the $x$-axis from below at $v_1 < v_2$.  If no such $v_1$ exists, let $v_1 = 0$.  Then it follows that $f(v_1)-f^{\alpha}(v_1) \geq 0$.  This is true even if $v_1 = 0$, as in this case, for all $v$ in the interval $[v_1, v_2]$, $F(v)-F^{\alpha}(v) \geq 0$.

Let $q(v)$ be the quantile of $v$ in $F$, and let $q^{\alpha}(v)$ be the quantile of $v$ in $F^{\alpha}$.  Note that $q(v_1) = q^{\alpha}(v_1)$ and $q(v_3) = q^{\alpha}(v_3)$.  By Lemma~\ref{densitybound}, for all $v \geq v_1$, \[f(v) \geq f(v_1)\left(\frac{q(v)}{q(v_1)}\right)^{2-\alpha}.\]  Because $f(v_1) \geq f^{\alpha}(v_1)$ and $q(v_1) = q^{\alpha}(v_1)$, the above is bounded below by \[f^{\alpha}(v_1)\left(\frac{q(v)}{q^{\alpha}(v_1)}\right)^{2-\alpha}.\]  On setting $v = v_3$, as $q(v_3) = q^{\alpha}(v_3),$ we obtain the bound \[f(v_3) \geq f^{\alpha}(v_1)\left(\frac{q^{\alpha}(v_3)}{q^{\alpha}(v_1)}\right)^{2-\alpha}.\]  However, the right-hand side is equal to $f^{\alpha}(v_3)$, which contradicts the statement that $f(v_3)- f^{\alpha}(v_3) < 0$.
\end{proof}

\begin{lem} \label{square}
Let $0 < \alpha < 1$ and let $F$ be an $\alpha$-SR distribution.  Then 
\[ 
\frac{\alpha}{1+\alpha} 
\int_0^{\infty} (1-F(v))dv \leq \int_0^{\infty} (1-F(v))^2dv.\] 
\end{lem}
\begin{proof} 
We start by defining the distribution $G$ by rescaling $F$'s argument so that $\int_0^{\infty} (1-G(v))dv = 1$.  Let \[G(v) = F\left(v\cdot \left(\int_0^{\infty} (1-F(w))dw\right)^{-1}\right) = F(v\cdot \lambda)\] where $\lambda = \left(\int_0^{\infty} (1-F(w))dw\right)^{-1}$.  Note that $\int_0^{\infty} (1-G(v))dv = \int_0^{\infty} (1-F(\lambda v))dv = \int_0^{\infty} (1-F(w))\lambda dw = 1$.  As $G$ is obtained by rescaling $F$'s argument, it is easy to see that $G$ is also $\alpha$-SR, and that \[\frac{\int_{0}^{\infty} (1-G(v))dv}{\int_{0}^{\infty} (1-G(v))^2dv} = \frac{\int_0^{\infty} (1-F(v))dv}{\int_0^{\infty} (1-F(v))^2dv}.\]  Therefore proving the lemma for $G$
implies the lemma for $F$. 


 
Let $G^{\alpha}$ be defined analogously with respect to the worst case distribution $F^{\alpha}$.  A straightforward calculation shows that the distribution $G^{\alpha}$ satisfies the inequality in the lemma.  
Therefore it is enough to prove that \[\frac{\int_{0}^{\infty} (1-G(v))dv}{\int_{0}^{\infty} (1-G(v))^2dv} \leq \frac{\int_{0}^{\infty} (1-G^{\alpha}(v))dv}{\int_{0}^{\infty} (1-G^{\alpha}(v))^2dv} 
\numberthis \label{ratio}.\] 

As both $G$ and $G^{\alpha}$ are normalized so that $\int_{0}^{\infty} (1-G(v))dv = \int_{0}^{\infty} (1-G^{\alpha}(v))dv = 1$, 
we can show~\eqref{ratio}, and consequently the lemma, by showing that \[\int_0^{\infty} (1-G(v))^2dv \geq \int_0^{\infty} (1-G^{\alpha}(v))^2dv,\] i.e. that \[\int_0^{\infty} (1-G(v))^2dv-\int_0^{\infty} (1-G^{\alpha}(v))^2dv \geq 0\] or equivalently that \[\int_0^{\infty} [(1-G(v))-(1-G^{\alpha}(v))]\cdot [(1-G(v))+(1-G^{\alpha}(v))]dv \geq 0.\]

We apply Lemma~\ref{cdfbound} to $G$ and $G^{\alpha}$.  Because $G^{\alpha}$ is the normalized version of the worst case distribution, the conditions of Lemma~\ref{cdfbound} hold.  It follows that there exists a $v_0$ such that $G(v) \geq G^{\alpha}(v)$ when $v \geq v_0$, and $G(v) \leq G^{\alpha}(v)$ when $v < v_0$. (Possibly $v_0 = \infty.$)  

\hide{Because $\int_0^{\infty} (1-G(v))-(1-G^{\alpha}(v))dv = 0$, it follows that \[\int_0^{v_0} (1-G(v))-(1-G^{\alpha}(v))dv = -\int_{v_0}^{\infty} (1-G(v))-(1-G^{\alpha}(v))dv.\] } 

Both $1-G$ and $1-G^{\alpha}$ are decreasing functions and hence so is $(1-G)+(1-G^{\alpha})$.  Thus, \begin{align*}
\int_0^{\infty} &[(1-G(v))-(1-G^{\alpha}(v))]\cdot [(1-G(v))+(1-G^{\alpha}(v))]dv \\
=& \int_{v_0}^{\infty} [(1-G(v))-(1-G^{\alpha}(v))]\cdot [(1-G(v))+(1-G^{\alpha}(v))]dv \\
&+\int_0^{v_0} [(1-G(v))-(1-G^{\alpha}(v))]\cdot [(1-G(v))+(1-G^{\alpha}(v))]dv \\
\geq & [(1-G(v_0))+(1-G^{\alpha}(v_0))]\int_{v_0}^{\infty} [(1-G(v))-(1-G^{\alpha}(v))]dv \tag*{as $(1-G(v))-(1-G^{\alpha}(v)) \leq 0$ when $v \geq v_0$} \\
&+ [(1-G(v_0))+(1-G^{\alpha}(v_0))]\int_0^{v_0} [(1-G(v))-(1-G^{\alpha}(v))]dv \tag*{as $(1-G(v))-(1-G^{\alpha}(v)) \geq 0$ when $v < v_0$} \\
&= [(1-G(v_0))+(1-G^{\alpha}(v_0))]\int_{0}^{\infty} [(1-G(v))-(1-G^{\alpha}(v))]dv \\
&= 0 \tag*{as $\int_{0}^{\infty} [(1-G(v))-(1-G^{\alpha}(v))] = 0$.}
\end{align*}

\hide{As $(1-G(v))-(1-G^{\alpha}(v)) \leq 0$ when $v \geq v_0$, $\int_{v_0}^{\infty} [(1-G(v))-(1-G^{\alpha}(v))][(1-G(v))+(1-G^{\alpha}(v))]dv \geq [(1-G(v))+(1-G^{\alpha}(v))]\int_{v_0}^{\infty} [(1-G(v))-(1-G^{\alpha}(v))]dv$.  Similarly, as $(1-G(v))-(1-G^{\alpha}(v)) \geq 0$ when $v < v_0$, $\int_0^{v_0} [(1-G(v))-(1-G^{\alpha}(v))][(1-G(v))+(1-G^{\alpha}(v))]dv \geq [(1-G(v))+(1-G^{\alpha}(v))]\int_0^{v_0} [(1-G(v))-(1-G^{\alpha}(v))]dv.$ it follows that \begin{multline*}
\int_0^{\infty} ((1-G(v))-(1-G^{\alpha}(v)))((1-G(v))+(1-G^{\alpha}(v)))dv \geq \\
((1-G(v_0))+(1-G^{\alpha}(v_0)))\left(\int_0^{\infty} ((1-G(v))-(1-G^{\alpha}(v)))dv\right) \geq 0\end{multline*} as desired.}
\end{proof}



\begin{lem}\label{minmax}
Let $0 < \alpha < 1$ and let $F$ be an $\alpha$-SR distribution.  Then, \begin{align*}\mathbb{E}_{v_1, v_2 \sim F}[\max\{v_1, v_2\}] &\leq \left(\frac{2+\alpha}{\alpha}\right)\mathbb{E}_{v_1, v_2 \sim F}[\min\{v_1, v_2\}] \\ &= \left(\frac{2+\alpha}{\alpha}\right)\mathbb{E}_{v_1, v_2 \sim F}[\max\{\phi(v_1), \phi(v_2)\}].\end{align*}
\end{lem}
\begin{proof} 
We first note that the equality follows by Myerson's Lemma~\cite{M}.  We now prove the inequality.  Note that $\Pr[\max\{v_1, v_2\} \geq x] = 1-F(x)F(x)$.  Then \begin{align*}\mathbb{E}_{v_1, v_2 \sim F}[\max\{v_1, v_2\}] &= \int_0^{\infty} x\frac{d}{dx} [F(x)^2]dx \\ &= \int_0^{\infty} 1-F(x)F(x)dx~~~ \text{(on integrating by parts).}\end{align*}  Similarly, $\Pr[\min\{v_1, v_2\} \geq x] = (1-F(x))(1-F(x)).$ Thus \[\mathbb{E}_{v_1, v_2 \sim F}[\min\{v_1, v_2\}] =  \int_0^{\infty} (1-F(x))(1-F(x))dx.\]  Therefore,  \begin{multline*}\frac{\mathbb{E}_{v_1, v_2 \sim F}[\max\{v_1, v_2\}] }{\mathbb{E}_{v_1, v_2 \sim F}[\min\{v_1, v_2\}]} = \frac{\int_0^{\infty} 1-F(x)F(x)dx}{\int_0^{\infty} (1-F(x))(1-F(x)) dx} = \\ \frac{\int_0^{\infty} 2(1-F(x))- (1-F(x))^2 dx}{\int_0^{\infty} (1-F(x))^2 dx} = \frac{2\int_0^{\infty} (1-F(x)) dx}{\int_0^{\infty} (1-F(x))^2 dx}-1.\end{multline*}
By applying Lemma~\ref{square}, we see this is bounded above by \[\left(\frac{2(1+\alpha)}{\alpha}-1\right) = \frac{2+\alpha}{\alpha}  ~~\text{as desired.}\] 
\end{proof}


\begin{proof}[{\normalfont Proof of \textbf{Lemma~\ref{conditioned}}}]
%
By Lemma~\ref{minmax}, Lemma~\ref{conditioned} holds when $t = 0$.  We now prove the result for $t > 0$.

Let $C(\alpha) = \left(\frac{2+\alpha}{\alpha}\right)$, and note that as $F$ is regular, $\phi$ is increasing, and hence $\max(\phi(v_1),\phi(v_2)) = \phi(\max\{v_1, v_2\})$.  Then, by substituting $\max\{v_1, v_2\} - \frac{1}{h(\max\{v_1, v_2\}}$ for $\phi(\max\{v_1, v_2\})$ in the statement of Lemma~\ref{conditioned}, we see that it is equivalent to \[\mathbb{E}_{v_1, v_2 \sim F}\left[\left(C(\alpha)-1\right)\max\{v_1, v_2\} - \frac{C(\alpha)}{h(\max\{v_1, v_2\})} \; \biggr| \; \max\{v_1, v_2\} \geq t\right] \geq 0.\]  We rewrite this as \begin{multline*}C(\alpha) \cdot \mathbb{E}_{v_1, v_2 \sim F}\biggr[\left(1-\alpha\right)\max\{v_1, v_2\} - \frac{1}{h(\max\{v_1, v_2\})} + \\ \left(\alpha - \frac{1}{C(\alpha)}\right)\max\{v_1, v_2\} \; \biggr| \; \max\{v_1, v_2\} \geq t\biggr] \geq 0.\end{multline*}  As $\frac{d\phi}{dv} \geq \alpha$, $\frac{d(\phi-\alpha v)}{dv} \geq 0$, and consequently, $\left(1-\alpha\right)\max\{v_1, v_2\} - \frac{1}{h(\max\{v_1, v_2\})}$ is always non-decreasing as a function of $\max\{v_1, v_2\}$. Additionally, we note that $1/C(\alpha) \leq \alpha$.  Therefore, conditioning on the event that $\max\{v_1, v_2\} \geq t$ only increases the expected value.
\end{proof}

\subsection{Revenue of the VCG-L Mechanism}

The VCG-L mechanism, as defined in~\cite{DRY}, is used in settings in which each bidder has an attribute (a classification)
and for each attribute there is a corresponding known distribution from which the bidder's valuation is drawn. The VCG-L mechanism uses the reserve prices, one per bidder, as defined in Section~\ref{sec:prelims}, as follows.  First, the VCG mechanism is run.  Second, all bidders whose valuation is less than their reserve price are removed.  Finally, each winning bidder is charged the larger of its reserve price and its VCG payment from the first step.

 In~\cite{DRY} the expected revenue of the VCG-L mechanism on MHR distributions was shown to achieve a $1/e$ approximation of
the welfare, or efficiency, of the VCG mechanism, which is tight.
In Theorem~\ref{thm:VCGL} we extend the analysis to $\alpha$-distributions;
the bound is again tight, as shown by the case of a single bidder drawn from the worst-case distribution $F^{\alpha}$.
We note that  the mechanism does not achieve a constant factor approximation in the case of regular distributions~\cite{DRY}.
\hide{\rjc{I think we could apply the sampling to this result.}}

\begin{thm}
\label{thm:VCGL}
For every downward-closed environment with valuations drawn independently from $\alpha$-SR distributions where $0< \alpha < 1$, the expected revenue of the VCG-L mechanism with monopoly reserves is at least an ${\alpha^{1/(1-\alpha)}}$ fraction of the expected efficiency of the VCG mechanism.
\end{thm}

\begin{proof}
Lemma~\ref{welfare} below replaces Lemma 3.10 in the proof of Theorem 3.11 in~\cite{DRY}.
The rest of the proof is unchanged.
\end{proof}

The proof of Lemma~\ref{welfare} uses the fact that $(\alpha+1)/\alpha \leq \alpha^{-1/(1-\alpha)}$, shown in Lemma~\ref{alphabound}, 
 and lower and upper bounds on the hazard rate $h(v)$, given in
 Lemmas~\ref{hbound} and~\ref{hbound2}, respectively.

\begin{lem}\label{alphabound}
For $0 < \alpha < 1$, $(\alpha+1)/\alpha \leq \alpha^{-1/(1-\alpha)}$.
\end{lem}

\begin{proof}
By rearranging the terms, we see that proving the lemma is equivalent to proving that $(\alpha+1)^{1-\alpha} \leq (1/\alpha)^{\alpha}$.  We replace $\alpha$ with $1/x$, and therefore it is enough to prove that for $x > 1$, \[\left(\frac{1}{x}+1\right)^{1-1/x} \leq x^{1/x}.\]  Again, by rearranging terms, it is enough to show that \[\left(\frac{x+1}{x}\right)^x = \left(1+\frac{1}{x}\right)^x \leq 1+x.\]  The left-hand side is at most $e$, and therefore the inequality is true when $x \geq e-1$.

When $x < e-1$, using the power series expansion for the left-hand side, we can bound it by $1+1+(x-1)/(2x) = 5/2-1/(2x)$.  The right-hand side is bounded above by $1+x$ if and only if $3x-1 \leq 2x^2$, which holds when $x > 1$, as desired.
\end{proof}


In the proof of Lemma~\ref{welfare}, and other lemmas, we often refer to the cumulative hazard rate, $H(v) = \int_0^v h(x)dx$.  We can relate $F$ and $H$ by the following identity, which follows by differentiating $\ln(1-F(v))$.  \[1-F(v) = e^{-H(v)}. \numberthis \label{hazard}\]







The following lemma gives a lower bound on $h(v)$ which will be used in Lemma~\ref{welfare}.
\begin{lem}\label{hbound}
Let $0 \leq \alpha \leq 1$ and let $F$ be an $\alpha$-SR distribution with virtual valuation function $\phi$.  Then for all $v_1 \leq v_2$, \[\frac{1}{(1-\alpha)(v_2-v_1)+1/h(v_1)} \leq h(v_2).\] 
\end{lem}
\begin{proof}
When $\alpha = 1$, this states that $h(v_2) \geq h(v_1)$ (for $\phi(v_2) - \phi(v_1) = v_2 - 1/h(v_2)-(v_1-1/h(v_1)) \geq v_2-v_1 $ in this case).

By definition, as $\phi$ is $\alpha$-SR, $\phi(v_2)-\phi(v_1) \geq \alpha(v_2-v_1)$.  
Substituting $\phi(v) = v-1/h(v)$ yields 
\[\left(v_2-\frac{1}{h(v_2)}\right)-\left(v_1-\frac{1}{h(v_1)}\right) \geq \alpha(v_2-v_1)\] i.e.
\[(1-\alpha)(v_2-v_1)+\frac{1}{h(v_1)} \geq \frac{1}{h(v_2)},\]  
from which the desired inequality follows.
\end{proof}

Using almost the same proof as above, we obtain the following upper bound on $h(v)$, also used in Lemma~\ref{welfare}.  This result was also shown in~\cite{CR1} for the special case of $v_2 = r$.
\begin{lem}\label{hbound2}
Let $0 \leq \alpha \leq 1$ and let $F$ be an $\alpha$-SR distribution with virtual valuation function $\phi$.  Then for all $v_1 \leq v_2$ such that $1/h(v_2)-(1-\alpha)(v_2-v_1) > 0$, \[h(v_1) \leq \frac{1}{1/h(v_2)-(1-\alpha)(v_2-v_1)}.\] 
\end{lem}
\begin{proof}
Again, when $\alpha = 1$, this states that $h(v_1) \leq h(v_2)$.

As in the proof of Lemma~\ref{hbound}, \[\frac{1}{h(v_1)} \geq \frac{1}{h(v_2)}-(1-\alpha)(v_2-v_1).\]  If $1/h(v_2)-(1-\alpha)(v_2-v_1) > 0$, then taking the reciprocal of both sides yields the desired inequality.
\end{proof}

Note that for continuous distributions, the condition $1/h(v_2)-(1-\alpha)(v_2-v_1) > 0$ holds when $v_2 = r$, where $r$ is the reserve price, as $1/h(r) = r$. Also note that Lemmas~\ref{hbound} and~\ref{hbound2} hold in the case that $F$ is defined on a discrete set.

We now state and prove Lemma~\ref{welfare}.

\begin{lem}\label{welfare}
Let $0 < \alpha < 1$ and let $F$ be an $\alpha$-SR distribution, with monopoly price $r$ and revenue function $\hat{R}$.  Let $V(t)$ denote the expected welfare of a single-item auction with a posted price of t and a single bidder with valuation drawn from $F$. For every non-negative number $t \geq 0$, \[\hat{R}(\max\{t, r\}) \geq \alpha^{1/(1-\alpha)} V(t).\]
\end{lem}
\begin{proof}
As in the proof in~\cite{DRY} of the corresponding lemma for MHR distributions, we split this into two cases, $t \leq r$ and $t \geq r$.  In both cases, we can write the left-hand side as $s \cdot (1-F(s)) = s\cdot e^{-H(s)}$, where $H(v) = \int_0^{v} h(v)$, and $s = \max\{t, r\}$.

\noindent\textbf{Case 1}: $t \leq r$.

\[\mbox{We start from the fact that}~~V(t) \leq \int_0^{\infty} e^{-H(v)}dv,\hspace*{1.1in} \numberthis \label{vt}\]  
as shown in Lemma 3.10 in~\cite{DRY} if $h$ is non-negative, which therefore still applies for the case of $\alpha$-SR distributions.

In order to upper bound $V(t)$, we start by lower bounding $H(v)$.  Because $h(v)$ is always non-negative, $H(v)$ is always non-negative.  When $v \leq r$, this will be the only lower bound we use.  Otherwise, we lower bound $H(v)$ using the lower bound for $h(v)$ from Lemma~\ref{hbound} when $v \geq r$.  In particular, if $v \geq r$, then
\begin{align*}
H(v) &= \int_0^v h(v)dv 
= \int_0^r h(v)dv + \int_r^v h(v)dv 
= H(r)+ \int_r^v h(v)dv \\
&\geq H(r) + \int_r^v \frac{1}{(1-\alpha)(v-r)+r}dv \tag{by Lemma~\ref{hbound}}\\
&= H(r)+ \frac{1}{1-\alpha}\ln\left((1-\alpha)v+\alpha r\right) \biggr|_r^v \\
&=  H(r) + \frac{1}{1-\alpha}\ln\left((1-\alpha)\frac{v}{r}+\alpha)\right).
\end{align*}
Therefore,
\begin{align*}
V(t) &\leq \int_0^{\infty} e^{-H(v)}dv 
= \int_0^{r} e^{-H(v)}dv + \int_r^{\infty} e^{-H(v)}dv \\
&\leq \int_0^{r} e^{-H(v)}dv+ e^{-H(r)}\int_r^{\infty} e^{-\frac{1}{1-\alpha}\ln((1-\alpha)\frac{v}{r}+\alpha))}dv \\
&= \int_0^{r} e^{-H(v)}dv + e^{-H(r)}\int_r^{\infty} [(1-\alpha)\frac{v}{r}+\alpha]^{-1/(1-\alpha)} dv \\
&= \int_0^{r} e^{-H(v)}dv -e^{-H(r)} \frac{1-\alpha}{\alpha} \cdot \frac{r}{1-\alpha} \cdot \left(\frac{1-\alpha}{r}v+\alpha\right)^{-\alpha/(1-\alpha)} \Biggr|_r^{\infty} \\
&= \int_0^{r} e^{-H(v)}dv - e^{-H(r)}\frac{r}{\alpha} \left(\frac{1-\alpha}{r}v+\alpha\right)^{-\alpha/(1-\alpha)} \Biggr|_r^{\infty} \\
&= \int_0^{r} e^{-H(v)}dv+ e^{-H(r)}\frac{r}{\alpha} \numberthis \label{Vbound}
\end{align*}
\[\mbox{We rewrite this as}~~ \left(e^{H(r)} \int_0^{r} e^{-H(v)}dv+\frac{r}{\alpha}\right) e^{-H(r)}.\hspace*{0.5in}\]  
In order to upper bound this, we consider $e^{H(r)} \int_0^{r} e^{-H(v)}dv$ by itself.  Note that if $v\leq r$, on applying Lemma~\ref{hbound2} with $v_1 = v$ and $v_2 = r$, \[h(v) \leq \frac{1}{r-(1-\alpha)(r-v)} = \frac{1}{(1-\alpha) v+\alpha r}.\] 
\begin{align*}
\mbox{ It follows that}~~H(r)-H(v) = \int_v^r h(v)dv 
&\leq \int_v^r \frac{1}{(1-\alpha) v+\alpha r}dv \\
&= \frac{1}{1-\alpha}\ln\left((1-\alpha)v+\alpha r\right) \biggr |_v^r \\
&=  \frac{1}{1-\alpha}\ln\left(\frac{r}{(1-\alpha)v+\alpha r}\right).
\end{align*}
\begin{align*}
\mbox{Therefore,}~~e^{H(r)} \int_0^{r} e^{-H(v)}dv &\leq \int_0^r e^{\frac{1}{1-\alpha}\ln\left(\frac{r}{(1-\alpha)v+\alpha r}\right)}dv \hspace*{0.9in}\\
&= \int_0^r \left(\frac{(1-\alpha)v}{r}+\alpha \right)^{-1/(1-\alpha)}dv \\
&=  - \frac{r}{\alpha} \left(\frac{1-\alpha}{r}v+\alpha\right)^{-\alpha/(1-\alpha)} \Biggr|_0^{r} \\
&= \frac{r}{\alpha}\left(\alpha^{-\alpha/(1-\alpha)}-1\right).
\end{align*}
Plugging this into our bound for $V(t)$ yields
\begin{align*}
V(t) &\leq  \left(\frac{r}{\alpha}\left(\alpha^{-\alpha/(1-\alpha)}-1\right)+\frac{r}{\alpha}\right)e^{-H(r)}\\
&=  \alpha^{-1/(1-\alpha)} \cdot r \cdot e^{-H(r)} \\
&= \hat{R}(r) \alpha^{-1/(1-\alpha)}, \tag{as $\hat{R}(r) = re^{-H(r)}$ by~\eqref{hazard}}
\end{align*}  and as $\max\{t, r\} = r$ in this case, $\hat{R}(\max\{t, r\}) \geq \alpha^{1/(1-\alpha)} V(t)$ as desired.

\vspace{3mm}

\noindent\textbf{Case 2:} $t \geq r$.

From the proof of Lemma 3.10 in~\cite{DRY}, \[V(t) = e^{-H(t)} \cdot \left[ t + \int_t^{\infty} e^{-(H(v)-H(t)) } dv\right].\]  \begin{align*}
\text{For }v \geq t\text{, } H(v)-H(t) = \int_t^v h(v)dv &\geq \int_t^v \frac{1}{(1-\alpha)(v-r)+r}dv  \tag{by Lemma~\ref{hbound} with $v_2 = v, v_1 = r$} \\
&= \frac{1}{1-\alpha}\ln\left((1-\alpha)v+\alpha r\right) \biggr|_t^v \\
&=  \frac{1}{1-\alpha}\ln\left(\frac{(1-\alpha)v+\alpha r} {(1-\alpha)t+\alpha r}\right).
\end{align*}
Therefore,
\begin{align*}
\int_t^{\infty} e^{-(H(v)-H(t)) } dv &\leq \int_t^{\infty} e^{\frac{-1}{1-\alpha}\ln\left(\frac{(1-\alpha)v+\alpha r} {(1-\alpha)t+\alpha r}\right)}dv \\
&= \int_t^{\infty} \left(\frac{(1-\alpha)v+\alpha r} {(1-\alpha)t+\alpha r}\right) ^ {-1 / (1-\alpha ) }dv \\
&= \frac{-(1-\alpha)}{\alpha} \cdot \frac{(1-\alpha)t+\alpha r}{1-\alpha} \cdot \left(\frac{(1-\alpha)v+\alpha r} {(1-\alpha)t+\alpha r}\right) ^ {-\alpha / (1-\alpha) } \Biggr|_t^{\infty} \\
&= \frac{(1-\alpha)t+\alpha r}{\alpha} 
 \leq \frac{t}{\alpha} \numberthis \label{integralbound}.
\end{align*}
\begin{align*}
\mbox{It follows that}~~V(t) &= e^{-H(t)} \cdot \left[ t + \int_t^{\infty} e^{-(H(v)-H(t)) } dv\right] \\
&\leq e^{-H(t)} \cdot \left[ t + \frac{t}{\alpha} \right ] 
= e^{-H(t)} \cdot t \cdot \left(\frac{\alpha+1}{\alpha}\right) \numberthis \label{vtbound} \\
&\leq e^{-H(t)} \cdot t \cdot\alpha^{-1/(1-\alpha)}  \tag{by Lemma~\ref{alphabound}}\\
&= \hat{R}(t)\alpha^{-1/(1-\alpha)}.
\end{align*}
\end{proof}

\subsection{Single-item Auctions with Known Budgets}

Theorem~\ref{thm:two-mech} describes a mechanism that achieves a constant-factor approximation to
the social welfare, with the constraint that all bidders have a budget that is known.  As $\alpha$ tends to $1$, the approximation factor tends to that given originally in~\cite{CMM}.  However, it is not known if the given mechanism is optimal,
even for MHR distributions.

First, we generalize Lemma 3.1 in~\cite{HR} from the case $\alpha = 1$, which was also used to prove the original theorem in~\cite{CMM}.  It upper bounds values $v$ in terms of the monopoly price $r$ and the virtual valuation $\phi(\alpha)$.
%
\begin{lem}\label{reservevirtualbound}
Let $\alpha \geq 0$ and let $F$ be an $\alpha$-SR distribution with monopoly price $r$ and virtual valuation function $\phi$.  Then, if $v \geq r$, \[v \leq r+\frac{\phi(v)}{\alpha}.\]
\end{lem}
\begin{proof}
By definition, as $\phi$ is $\alpha$-SR, $\phi(v)-\phi(r) \geq \alpha(v-r)$.  
Because $r$ is the monopoly price, $\phi(r) \geq 0$; thus $\phi(v) \geq \alpha(v-r)$.
Solving for $v$ gives the above inequality.
\end{proof}

\begin{thm}
\label{thm:two-mech}
Let $0 < \alpha < 1$.  For every downward-closed environment with valuations drawn independently from distributions that are $\alpha$-SR, and with budgets $B_i$ for each bidder $i$, the mechanism that chooses each of the following two mechanisms with probability one half
gives a $\frac{4}{\alpha}+2\frac{\alpha+1}{\alpha^{(2-\alpha)/(1-\alpha)}}$-approximation to the social welfare of a welfare-optimal budget-feasible mechanism.
{The resulting mechanism is}
dominant-strategy incentive-compatible, ex-post individually rational, and budget feasible.
\begin{itemize}
\item Mechanism 1: Always allocate to the set $S_1^*$ and charge zero payments, where $S_1^* = \argmax \sum_{i \in S} v_i^*$ and $v_i^*$ is the reserve price for the $i$th bidder.
\item Mechanism 2: Elicit values from the agents; for all $i$ with $v_i > B_i$, replace $v_i$ with $B_i$; run Myerson's mechanism on the resulting instance.
\end{itemize}
\end{thm}

\begin{proof}
We start by upper-bounding the social welfare of any allocation $x$.  First note that Lemma~\ref{reservevirtualbound} implies that $v_i \leq v_i^* + \phi(v_i)/\alpha + \phi^-(v_i)/\alpha$ for all $v_i$ where $\phi^-(v_i)$ is equal to $-\phi(v_i)$ if $v_i \leq v_i^*$, and $0$ otherwise.  Additionally, by Lemma 20 in~\cite{CMM}, $\int \phi^-(v_i)x(v_i)dF(v_i) \leq \int v_i^*x(v_i)dF(v_i)$.  Combining these yields \begin{align*}
&\int_{\mathbf{v}} \left(\sum_i v_ix_i(\mathbf{v})\right)dF(\mathbf{v}) \leq \int_{\mathbf{v}} \left(\sum_i \left(\frac{\phi_i(v_i)}{\alpha}+\left(1+\frac{1}{\alpha}\right)v_i^*\right)x_i(\mathbf{v})\right)dF(\mathbf{v}) \\
&~~~~~~~~=  \frac{1}{\alpha}\int_{\mathbf{v}} \left(\sum_i \left(\phi_i(v_i)\right)x_i(\mathbf{v})\right)dF(\mathbf{v})+\left(\frac{\alpha+1}{\alpha}\right)\int_v \left(\sum_i \left(v_i^*\right)x_i(\mathbf{v})\right)dF(\mathbf{v}) 
\numberthis \label{CMMbound} \end{align*} as an upper bound on the social welfare.

We now lower bound the social welfare of the two mechanisms.  The social welfare of Mechanism $1$ can be lower bounded by $\sum_{i \in S^*} \mathbb{E}[v_i] \geq \sum_{i \in S^*} (1-F(v_i^*))v_i^* \geq \alpha^{1/(1-\alpha)} \sum_{i \in S^*} v_i^* \geq \alpha^{1/(1-\alpha)} \int_v \sum_i v_i^*x(v)dF(v)$.  The second inequality follows from Lemma~\ref{reservebound} below.   By Theorem 7 in~\cite{CMM}, Mechanism $2$ achieves social welfare at least $\frac{1}{2} \int_v \sum_i \phi_i(v_i) x(v)dF(v)$.

Let $W_1$ be the social welfare of Mechanism $1$ and $W_2$ the social welfare of Mechanism $2$.  Then we can upper-bound~\eqref{CMMbound} by \[\frac{2}{\alpha}W_2+\frac{(\alpha+1)}{\alpha^{1+1/(1-\alpha)}}W_1 = \frac{2}{\alpha}W_2+\frac{(\alpha+1)}{\alpha^{(2-\alpha)/(1-\alpha)}}W_1 \numberthis \label{twomechbound}.\]  Let Mechanism $i$ be the mechanism with the greater social welfare.  Therefore,~\eqref{twomechbound} is bounded above by \[\left(\frac{2}{\alpha}+\frac{(\alpha+1)}{\alpha^{(2-\alpha)/(1-\alpha)}}\right)W_i.\] Finally, the mechanism stated in the lemma achieves social welfare at least $W_i/2$, which yields the desired approximation factor.
\end{proof}

\subsection{Single-item Auctions with Private Budgets}

The following theorem describes a mechanism that achieves a constant-factor approximation to
the revenue, with the constraint that all bidders have a budget that is unknown.  As $\alpha$ tends to $1$, the approximation factor tends to that given originally in~\cite{CMM}.  However, it is not known if the given mechanism is optimal,
even for MHR distributions.

\begin{thm}\label{lottery}
Let $0 < \alpha < 1$.  For every downward-closed, single-parameter environment with valuations drawn independently from distributions that are $\alpha$-SR, and with private budgets $B_i$ drawn from a known distribution for each bidder $i$, there exists a mechanism, given in~\cite{CMM}, that gives a $3\left(1+\alpha^{^{-1/(1-\alpha)}}\right)$-approximation to the optimal revenue.
\end{thm}
\begin{proof}
In the proof of Theorem 14 in~\cite{CMM}, substitute the bound
from Lemma~\ref{reservebound} below;
it lower bounds the probability that a valuation exceeds the monopoly price.
\end{proof}

\begin{lem}\label{reservebound}\cite{CR1}
Let $F$ be an $\alpha$-SR distribution with monopoly price $r$.
%
%
If $0 < \alpha < 1$, then $1-F(r) \geq \alpha^{1/(1-\alpha)}$.
\end{lem}

\subsection{Multi-item auctions with Public Budgets}
Finally, we consider a setting with discrete valuations, drawn from the set $\{1, \ldots, L\}$,
and with public budgets, as defined in~\cite{BGGM}.
In Theorem~\ref{lpmech}, we show that the mechanism in~\cite{BGGM} achieves a
$\frac{192}{\alpha}\left(\frac{2-\alpha}{\alpha}\right)^{1/(1-\alpha)}$-approximation,
which matches the bound given in~\cite{BGGM}
as $\alpha$ tends to 1.

\begin{thm}\label{lpmech}
Let $0 < \alpha < 1$.  Consider the setting with multiple bidders and items where each bidder $i$  has a public budget $B_i$, and the valuation of each item by each bidder is drawn independently from an $\alpha$-SR distribution.  Then there exists a mechanism that is universally truthful dominant-strategy incentive-compatible, and gives a $\frac{192}{\alpha}\left(\frac{2-\alpha}{\alpha}\right)^{1/(1-\alpha)}$-approximation to the revenue of the optimal truthful-in-expectation Bayesian incentive-compatible mechanism.
\end{thm}
\begin{proof}
Lemma~\ref{lem:pot-large} below replaces Lemma~13 in the analysis in~\cite{BGGM}.
\end{proof}

\begin{lem}\label{lem:pot-large}
Let $0 < \alpha < 1$ and let $F$ be an $\alpha$-SR distribution defined on the set $\{1, \ldots, L\}$ with monopoly price $r$ and virtual valuation function $\phi$.  Then \[\Pr_{v \sim F} [\phi(v) > \alpha v/2 ] \geq \left(\frac{\alpha}{2-\alpha}\right)^{1/(1-\alpha)}.\]
\end{lem}

\begin{proof}
As in~\cite{BGGM}, we construct a probability distribution with support $[0, L]$ that approximates $F$.  In particular, let $\hat{F}$ be such that \[F(v)-F(v-1) = \int_{v-1}^{v}\hat{f}(v)dv\] for all integers $v$, where $\hat{F}$ is the cumulative distribution function of $\hat{f}$.  Such a distribution can be constructed by letting $\hat{f}(v) = F(\lceil v \rceil)-F(\lceil v \rceil-1)$.  Let $\hat{h}$ be the hazard rate of $\hat{F}$.

It was shown in~\cite{BGGM} that $\int_{r-1}^r \hat{h}(v)dv \leq h(r)$ for integers $r \in \{1, \ldots, L\}$, and therefore that \[\int_0^r \hat{h}(v)dv \leq \sum_{v=1}^{r} h(v).\]  Let $k^*$ be the integer such that $\phi(v) \geq \alpha v/2$ if and only if $v > k^*$. Then $k^*-1/h(k^*) \leq \alpha k^*/2$ which on rearranging yields $1/h(k^*) \geq ((2-\alpha)k^*)/2$ and, 
\begin{align*}1/h(k^*)-(1-\alpha)(k^*-v) &\geq (2-\alpha)k^*/2-(1-\alpha)(k^*-v )\\
&= \frac{\alpha}{2} k^*+(1-\alpha)v \geq 0.
\end{align*} 
Therefore the condition in Lemma~\ref{hbound2} holds.  This implies that for all $v \leq k^*$, \[h(v) \leq \frac{1}{(1-\alpha)(v-k^*)+(2-\alpha) k^* / 2},\] and thus
\begin{align*}
\sum_{i=1}^{k^*} h(v) &\leq \sum_{i=1}^{k^*} \frac{1}{(1-\alpha)(v-k^*)+(2-\alpha) k^* / 2} \\
&\leq \int_0^{k^*} \frac{1}{(1-\alpha)(v-k^*)+(2-\alpha) k^* / 2}dv \\
&= \frac{1}{1-\alpha}\log((1-\alpha)(v-k^*)+(2-\alpha) k^* / 2)  \biggr|_0^{k^*}\\
&= \frac{1}{1-\alpha} \log\left(\frac{(2-\alpha) k^* / 2}{(2-\alpha) k^* / 2 -(1-\alpha) k^*}\right) \\
&= \frac{1}{1-\alpha} \log\left(\frac{2-\alpha}{\alpha}\right).
\end{align*} where the second inequality uses the fact that $\frac{1}{(1-\alpha)(v-k^*)+(2-\alpha) k^* / 2}$ is decreasing.

Finally, as in the proof of Lemma 3.3 in~\cite{BGGM},
\begin{align*}
\Pr_{v \sim F} [\phi(v) > \alpha v/2 ] &= e^{-\int_0^{k^*} \hat{h}(v)dv} 
\geq e^{\frac{-1}{1-\alpha} \log\left(\frac{2-\alpha}{\alpha}\right)} \\
&= \left(\frac{2-\alpha}{\alpha}\right)^{-1/(1-\alpha)}~~~\mbox{as desired.}
\end{align*} 
\end{proof}

\section{Sample Complexity}
\label{sec:sampling}

Here we discuss the sample complexity of the mechanisms used in Theorems~\ref{thm:vcgl-samp}--\ref{thm-BGGm-w-sampling}.
These mechanisms require knowledge of the distributions from which the valuations are drawn.  
We will show how to modify these mechanisms when the distributions are known approximately via samples.

In~\cite{CR1}, this was done for the single-item auction by modifying the Myerson Auction.  Our mechanisms will follow the same format.  After obtaining the samples, we estimate the distributions, and use these in place of the actual distribution.  In particular, given $m$ samples, we first discard the $\lfloor \xi m \rfloor-1$ largest samples.  We let the empirical quantile of the $j$th largest sample $v_j$ be $\overline{q}(v_j) = 1-\overline{F}(v_j) = \frac{2j-1}{2m}$ and we let the empirical revenue curve be defined as $\overline{R}\left(\frac{2j-1}{2m}\right) = \frac{2j-1}{2m} v_j$, $\overline{R}(0) = \overline{R}(1) = 0$, with straight lines joining successive points.  Because we will need to associate an empirical quantile with each value smaller than the $(\lfloor \xi m \rfloor-1)$th sample, we use the empirical revenue curve to define these quantiles, as follows.  We define $\overline{v}(\overline{q})\cdot\overline{q} = \overline{R}(\overline{q})$, and $\overline{q}(v) = \overline{v}^{-1}(v)$.

We let $\CR$ be the convex hull of the actual revenue curve (which is convex for regular distributions).  Also, we let $\overline{\CR}$ be the convex hull of the empirical revenue curve, $\overline{\phi}$ the empirical virtual valuation function (i.e.\ the slope of $\overline{\CR}$), and $\overline{r}$ the empirical reserve price, i.e.\ the largest value such that $\overline{\phi}(\overline{r}) = 0. $ We will be overloading notation, writing both $\CR(q)$ and $\CR(v(q))$, and likewise for $\overline{\CR}$.


We start by stating results about the empirical revenue curve that will be used to modify mechanisms for the empirical setting.  The following is Lemma 6.2 from~\cite{CR1} and gives a lower bound on the accuracy of samples.

\begin{lem}\label{empquantrange}
Let $F$ be a regular distribution (not necessarily strongly regular).  Suppose $m$ independent samples with values $v_1 \geq v_2 \geq \cdots \geq v_m$ are drawn from $F$.  Let $0 < \gamma, \xi < 1$ be given.  Then for all $v \leq v_{\lfloor \xi m \rfloor}$ \[q(v) \in [\overline{q}(v)/(1+\gamma)^2, \overline{q}(v)(1+\gamma)^2] \numberthis \label{quantrange}\] with probability at least $1-\delta$, if $\gamma \xi m \geq 4$, $(1+\gamma)^2 \leq \frac{3}{2}$, and $m \geq \frac{3}{\gamma^2(1+\gamma)\xi}\max\{\frac{\ln{3}}{\gamma}, \ln{\frac{3}{\delta}}\}$.
\end{lem}

We use Lemma~\ref{empquantrange} to prove the following result which shows that the value of the revenue curve at the empirical reserve price is a good approximation to the revenue at the actual reserve price.  Let $\mathcal{E}_{\alpha}$ be the event that outcome~\eqref{quantrange} occurs for all $v \leq v_{\lfloor \xi m \rfloor}.$  Note that $\mathcal{E}_{\alpha}$ holds with probability at least $1-\delta$.  Let $\overline{\xi}$ be the empirical quantile of the largest retained sample: $\overline{\xi} = \frac{\lfloor 2 \xi m \rfloor-1}{2m} \geq \xi-\frac{1}{m}$.  Recall that $\overline{v}(\overline{\xi})$ denotes the value of this sample, and $\overline{q}(\overline{v}(\overline{\xi})) = \overline{\xi}$, of course.

\begin{lem}~\label{empreserve}
Conditioned on $\mathcal{E}_{\alpha}$, \[\CR(\overline{r}) \geq \frac{1-\overline{\xi}(1+\gamma)^2}{(1+\gamma)^4}\CR(r).\]
\end{lem}

\begin{proof}
We assume that the statement in Lemma~\ref{empquantrange} holds, which it does with probability $1-\delta$.  By Lemma~\ref{empquantrange}, \[\CR(\overline{r}) = q(\overline{r})\overline{r} \geq \frac{1}{(1+\gamma)^2}\overline{q}(\overline{r})\overline{r} = \frac{1}{(1+\gamma)^2}\overline{\CR}(\overline{r}). \numberthis \label{emp}\]  At this point we consider two cases, $\overline{q}(r) \geq \overline{\xi}$ and $\overline{q}(r) < \overline{\xi}$.

\noindent\textbf{Case 1:} $\overline{q}(r) \geq \overline{\xi}$.

By the definition of $\overline{r}$, $\overline{\CR}(\overline{r}) \geq \overline{\CR}(r)$ and therefore~\eqref{emp} is bounded below by $\frac{1}{(1+\gamma)^2}\overline{\CR}(r)$. We use Lemma~\ref{empquantrange} again to see that \[\frac{1}{(1+\gamma)^2}\overline{\CR}(r) = \frac{1}{(1+\gamma)^2} \overline{q}(r)r \geq \frac{1}{(1+\gamma)^4}q(r)r = \frac{1}{(1+\gamma)^4}\CR(r).\]

\noindent\textbf{Case 2:} $\overline{q}(r) < \overline{\xi}$.

By the definition of $\overline{r}$, $\overline{\CR}(\overline{r}) \geq \overline{\CR}(\overline{q}(\overline{v}(\overline{\xi})))$ and therefore~\eqref{emp} is bounded below by $\frac{1}{(1+\gamma)^2}\overline{\CR}(\overline{q}(\overline{v}(\overline{\xi})))$.  We use Lemma~\ref{empquantrange} again to see that \[\frac{1}{(1+\gamma)^2}\overline{\CR}(\overline{q}(\overline{v}(\overline{\xi}))) = \frac{1}{(1+\gamma)^2}\overline{q}(\overline{v}(\overline{\xi}))\overline{v}(\overline{\xi}) \geq \frac{1}{(1+\gamma)^4}q(\overline{v}(\overline{\xi}))\overline{v}(\overline{\xi})  = \frac{1}{(1+\gamma)^4}\CR(\overline{v}(\overline{\xi})).\]  Because $\CR$ is convex and by assumption $\overline{v}(\overline{\xi}) < r$, it follows that \[\frac{1}{(1+\gamma)^4}\CR(\overline{v}(\overline{\xi})) \geq \frac{1-q(\overline{v}(\overline{\xi}))}{(1+\gamma)^4}\CR(r).\]  Finally, by Lemma~\ref{empquantrange}, $q(\overline{v}(\overline{\xi})) \leq \overline{q}(\overline{v}(\overline{\xi}))(1+\gamma)^2$ and therefore we obtain a final bound of $\frac{1-\overline{\xi}(1+\gamma)^2}{(1+\gamma)^4}\CR(r)$ as desired.
\end{proof}

\hide{The following corollary to the work in~\cite{CR1} shows that given an appropriate number of samples, the empirical reserve price is a $(1-\epsilon)$-optimal reserve price.  By this, we mean that $R(\overline{r}) \geq (1-\epsilon)R(r)$.  The proof follows the argument of Lemma $4.1$ from~\cite{DRY}.

\begin{cor}~\label{empreserve}Let $F$ be a regular distribution, and let $\gamma > 0, \xi = \frac{k}{m}$ for integers $k$ and $m$ be given be given.  Then the empirical reserve price $\overline{r}$ based on $m$ samples is a $(1-11\gamma)$-optimal reserve price with probability $1-\delta$ if $\gamma \xi m \geq 1$, $(1+\gamma)^2 \leq \frac{3}{2}$ and \[m \geq \frac{6}{\gamma^2(1+\gamma)\xi} \max\left\{\frac{\ln 3}{\gamma}, \ln{\frac{6}{\delta}}\right\}.\]
\end{cor}

We let $\mathcal{E}_{a}$ be the event that the outcome of Corollary~\ref{empreserve} holds.}

The following lemma from~\cite{CR1} describes the accuracy of the empirical revenue curves.

\begin{lem}\label{empquant}
Assume that $\mathcal{E}_{\alpha}$ holds.  Then for all empirical quantiles $\overline{q}(\overline{r}) \geq \overline{q} \geq \overline{\xi}$, \[\frac{1}{(1+\gamma)^3}\CR(\overline{q}(1+\gamma)^2) = \frac{1}{(1+\gamma)}\overline{q}v(\overline{q}(1+\gamma)^2) \leq \overline{\CR}(\overline{q}) \leq \overline{q}v\left(\frac{\overline{q}}{(1+\gamma)^2}\right) = (1+\gamma)^2\CR\left(\frac{\overline{q}}{(1+\gamma)^2}\right).\]  If $\overline{q}(\overline{r}) \leq \overline{q}$, then the right inequality holds, and the left inequality becomes \[\frac{1}{(1+\gamma)^3}\CR(\overline{q}(1+\gamma)^3) = \overline{q}v(\overline{q}(1+\gamma)^2) \leq \overline{\CR}(\overline{q}).\]
\end{lem}

\begin{proof}
The upper bound on $\overline{\CR}(\overline{q})$ follows from Lemma 6.3 in~\cite{CR1}.  To prove the lower bound, we choose $j$ so that $t_{j+1} = \frac{2(j+1)-1}{2m} \geq \overline{q} \geq \frac{2j-1}{2m} = t_j$.  As defined in~\cite{CR1}, the values $t_i$ denote the empirical quantiles of the samples.    As stated in~\cite{CR1}, $t_j(1+\gamma) \geq t_{j+1}$, for $t_j \geq \overline{\xi}$.  If $\overline{q} \leq \overline{q}(\overline{r})$, then \begin{align*}
\overline{\CR}(\overline{q}) \geq \overline{\CR}(t_j) &\geq t_j v(t_j(1+\gamma)^2) \tag{by Lemma 6.3 in~\cite{CR1}} \\
&\geq \frac{\overline{q}}{(1+\gamma)} v(\overline{q}(1+\gamma)^2) \\
&= \frac{1}{(1+\gamma)^3}\CR(\overline{q}(1+\gamma)^2).
\end{align*}  Otherwise, if $\overline{q} \geq \overline{q}(\overline{r})$, then, \begin{align*}
\overline{\CR}(\overline{q}) \geq \overline{\CR}(t_{j+1}) &\geq t_{j+1} v(t_{j+1}(1+\gamma)^2) \tag{by Lemma 6.3 in~\cite{CR1}} \\
&\geq \overline{q} v(\overline{q}(1+\gamma)^3) \\
&= \frac{1}{(1+\gamma)^3}\CR(\overline{q}(1+\gamma)^3).
\end{align*}
\end{proof}

\subsection{Revenue of the VCG-L Mechanism}

We start with the VCG-L mechanism as defined in~\cite{DRY}.  As described previously, the VCG-L mechanism runs VCG but with lazy reserve prices.  In particular, after running VCG, all bidders who bid less than their reserve price are removed, and the remaining bidders are charged the maximum of their reserve price and the VCG payment.  When only given access to samples, we use the empirical reserve price, rather than the actual reserve price.

By a simple application of Lemma~\ref{empreserve}, we obtain the following bound.

\begin{thm}
\label{thm:vcgl-samp}
The expected revenue of the empirical VCG-L mechanism with $k$ classes of bidders is at least an ${\alpha^{1/(1-\alpha)}}\cdot\frac{1-\xi(1+\gamma)^2(1-k\delta)}{(1+\gamma)^4}$ fraction of the expected efficiency of the VCG mechanism, given $m \geq \frac{6(1+\gamma)}{\gamma^2\xi} \max\left\{\frac{\ln 3}{\gamma}, \ln{\frac{3}{\delta}}\right\}$ samples from each class, with $\gamma \xi m \geq 4$ and $(1+\gamma)^2 \leq \frac{3}{2}$.
\end{thm}

\begin{proof}
We show that with probability $1-k\delta$ the empirical VCG-L mechanism achieves revenue at least a $\frac{1-\xi(1+\gamma)^2}{(1+\gamma)^4}$ fraction of the VCG-L mechanism with total access to the distributions.  The statement of the theorem then follows from Theorem~\ref{thm:VCGL}.

Assume that $\mathcal{E}_{\alpha}$ holds for the samples from each distribution, which it does with probability $1-k\delta$.  For each bidder $i$, fix the valuations of all other bidders to be $\mathbf{v}_{-i}$, and let $t_i$ be bidder $i$'s VCG threshold.  Let $r_i$ be the reserve price for bidder $i$, and let $\overline{r}_i$ be the empirical reserve price.  In each of the two mechanisms, bidder $i$ is charged either its VCG threshold $t_i$, or the corresponding reserve price ($r_i$ in the original mechanism, and $\overline{r}_i$ in the empirical mechanism.)  We show that in each of the four possible cases, the expected revenue from bidder $i$ (over bidder $i$'s possible valuations) in the empirical VCG-L mechanism is at least a $\frac{1-\xi(1+\gamma)^2}{(1+\gamma)^4}$ fraction of the expected revenue from bidder $i$ in the original VCG-L mechanism.  Taking the expectation over $\mathbf{v}_{-i}$ and summing over all bidders proves the theorem.

\noindent\textbf{Case 1}: $r_i, \overline{r}_i \leq t_i$.

Bidder $i$ is charged the same amount in both mechanisms, and the expected revenue from bidder $i$ is also the same.

\noindent\textbf{Case 2}: $r_i, \overline{r}_i \geq t_i$.

Bidder $i$ is charged $r_i$ in the VCG-L mechanism and is charged $\overline{r}_i$ in the empirical VCG-L mechanism.  By Lemma~\ref{empreserve}, the expected revenue in the empirical VCG-L mechanism is at least a $\frac{1-\xi(1+\gamma)^2}{(1+\gamma)^4}$ fraction of the expected revenue in the VCG-L mechanism.

\noindent\textbf{Case 3}: $r_i \leq t_i \leq \overline{r}_i$.

Bidder $i$ is charged $t_i$ in the VCG-L mechanism and is charged $\overline{r}_i$ in the empirical VCG-L mechanism. By the convexity of the revenue curve, the expected revenue in the VCG-L mechanism is less than $\CR(r_i)$.  It follows from Lemma~\ref{empreserve} that the expected revenue from bidder $i$ in the empirical VCG-L mechanism is at least a $\frac{1-\xi(1+\gamma)^2}{(1+\gamma)^4}$ fraction of $\CR(r_i)$.

\noindent\textbf{Case 4}: $\overline{r}_i \leq t_i \leq r_i$.

Bidder $i$ is charged $r_i$ in the VCG-L mechanism and is charged $t_i$ in the empirical VCG-L mechanism.  By the convexity of the revenue curve, the expected revenue in the empirical VCG-L mechanism is at least $\CR(\overline{r}_i)$.  It follows from Lemma~\ref{empreserve} that $\CR(\overline{r}_i)$ is at least a $\frac{1-\xi(1+\gamma)^2}{(1+\gamma)^4}$ fraction of  the expected revenue in the VCG-L mechanism.
\end{proof}

\subsection{Single-item Auctions with Known Budgets}

Consider the mechanism in Theorem~\ref{thm:two-mech}, first given in~\cite{CMM}.  Recall that this mechanism is actually composed of two mechanisms, each chosen with probability one half.  Mechanism $1$ does not need access to the distribution, while Mechanism $2$ does, as one of its steps is to run Myerson's mechanism.  The sample complexity of a suitably modified version of Myerson's mechanism was already studied in Theorem 6.9 in~\cite{CR1}.  In particular with $n$ bidders, the expected revenue in the sampling setting is at least a $(1-\epsilon)$ fraction of the expected revenue with total access to the distribution, if given $m = \Omega\left(\left(\frac{n^{10}}{\epsilon^7}\right)\ln^3 \frac{n}{\epsilon}\right)$ samples.  This leads to the following theorem.

\begin{thm}\label{thm:two-mech-samp}
The mechanism described in Theorem~\ref{thm:two-mech}, when given access to $m = \Omega\left(\left(\frac{n^{10}}{\epsilon^7}\right)\ln^3 \frac{n}{\epsilon}\right)$ samples per class of bidders, gives a $\frac{4}{\alpha}+2\frac{\alpha+1}{\alpha^{(2-\alpha)/(1-\alpha)}(1-\epsilon)(1-n\delta)}$-approximation to the social welfare of a welfare-optimal budget-feasible mechanism.
\end{thm}

\subsection{Single-item Auctions with Private Budgets}

We now consider the mechanism in Theorem~\ref{lottery}, first given in~\cite{CMM}.  Let $(\mathbf{F}, \mathcal{S}, \mathbf{G})$ be the setting where $\mathbf{F}$ is the set of distributions from which each bidder's valuation is drawn, $\mathcal{S}$ is a matroid set system of allowable sets, and $\mathbf{G}$ is the set of distributions from which each bidder's budget is drawn.  Let $B$ be the actual budgets.  Then we define $\mathcal{B}$ as \[\mathcal{B} = \argmax_{S \in \mathcal{S}} \sum_{i \in S} \min\{v_i, B_i\}\] where $\mathbf{v}$ is the vector of bidder valuations.  Finally, let $T_i$ be the threshold for $i$'s inclusion in $\mathcal{B}$, \[T_i = \min\left\{v' : \; i \in \mathcal{B} \text{ for } ((\mathbf{v_{-i}}, v'), (\mathbf{B_{-i}}, v'))\right\}.\]  The proposed mechanism uses a lottery system which we define as follows.

\begin{define}
A lottery system $\mathcal{L}(p, p')$ either sets the price of an item at $p$ if $p \geq p'/3$, or allows the bidder to choose a value $a$, $2p/p' \leq a \leq 2/3,$ and then purchase an item at a price of $ap'/2$ with probability $1/3+a$.
\end{define}

The mechanism is to offer each bidder $i$ the lottery system $(T_i, r_i),$ where $r_i$ is the reserve price for bidder $i$.

When only given access to samples, we instead offer each bidder $i$ the lottery system $(T_i, \overline{r}_i)$, where $\overline{r}_i$ is the empirical reserve price based on $m \geq \frac{6(1+\gamma)}{\gamma^2\xi} \max\left\{\frac{\ln 3}{\gamma}, \ln{\frac{3}{\delta}}\right\}$ samples for some $0 < \delta < 1$ and $ \gamma$ satisfying $\gamma \epsilon m \geq 4$ and $(1+\gamma)^2 \geq 3/2$.

Let $\mathcal{M}$ be the optimal mechanism, and let $\mathcal{M}^{\mathcal{L}}$ be the proposed mechanism.  Let $\mathcal{R}^{\mathcal{M} \cap \mathcal{B}}$ and $\mathcal{R}^{\mathcal{M} \backslash \mathcal{B}}$ be the revenue that $\mathcal{M}$ obtains from serving those bidders in $\mathcal{B}$ and those not in $\mathcal{B}$, respectively.  Additionally, let $\mathcal{R}^{\mathcal{M^\mathcal{L}}}$ be the revenue from the proposed mechanism.  The following lemma relating $\mathcal{R}^{\mathcal{M} \backslash \mathcal{B}}$ and $\mathcal{R}^{\mathcal{M^\mathcal{L}}}$ from~\cite{CMM} still holds even after modifying the lottery system offered to the bidder.

\begin{lem}\label{part1}
 $\mathcal{R}^{\mathcal{M} \backslash \mathcal{B}} \leq 3\mathcal{R}^{\mathcal{M^\mathcal{L}}}.$
\end{lem}

Before relating $\mathcal{R}^{\mathcal{M} \cap \mathcal{B}}$ and $\mathcal{R}^{\mathcal{M^\mathcal{L}}}$, we prove a sequence of lemmas. The first, a slight modification of a claim in~\cite{CR1}, upper bounds the revenue curve for quantiles less than $q(r)$.  This will allow us to lower bound $q(\overline{r})$.

\begin{lem}\label{reserveb}
Let $F$ be an $\alpha$-SR distribution with reserve price $r$.  For $q \leq q(r)$ and $0 < \alpha < 1$, \[\CR(q) \leq \CR(q(r)) \frac{1}{1-\alpha}\left( \left(\frac{q}{q(r)}\right)^{\alpha}-\alpha\frac{q}{q(r)}\right).\]
\end{lem}
\begin{proof}By letting $q_0 = r$ in Lemma 6.3 from~\cite{CR1} and using the fact that $q(r)/f(r) = r$, it follows that \[v(q) \leq r+\frac{r}{1-\alpha}\left[\left(\frac{q(r)}{q}\right)^{1-\alpha}-1\right].\]   \begin{flalign*} \text{Thus, }  \CR(q) &= q\cdot v(q) \\
  &\leq qr+q\frac{r}{1-\alpha}\left[\left(\frac{q(r)}{q}\right)^{1-\alpha}-1\right] \\
 &= \CR(q(r))\frac{q}{q(r)}\left[1+\frac{\left(\frac{q(r)}{q}\right)^{1-\alpha}-1}{1-\alpha}\right] \\
 &= \CR(q(r))\frac{1}{1-\alpha}\left( \left(\frac{q}{q(r)}\right)^{\alpha}-\alpha\frac{q}{q(r)}\right).\end{flalign*} 
\end{proof}

\hide{\begin{lem}
Let $F$ be an $\alpha$-SR distribution with monopoly price $r$.  For $q_1 \leq q_2 \leq q(r)$ and $0 < \alpha < 1$, \[CR(q_1) \leq CR(q_2)\frac{1}{1-\alpha}\left(\frac{q_1}{q_2}\right)^{\alpha}.\]
\end{lem}}

\hide{We note that the original lemma was proved in the case $q_2 = q(r)$.  However, the proof still holds in the case $q_2 \leq q(r)$.}

The following lemma lower bounds the quantile of the empirical reserve price in terms of the quantile of the actual reserve price.

\begin{lem}\label{reservequant}
Let $0 < \alpha \leq 1$, and let $F$ be an $\alpha$-SR distribution.  Then assuming that $\mathcal{E}_{\alpha}$ holds for the samples from each distribution, \[q(\overline{r}) \geq \left(1-\sqrt{\frac{8\gamma}{\alpha}}\right) q(r).\]
\end{lem}
\begin{proof} 
If $q(\overline{r}) \ge q(r)$ the result holds trivially; so for the rest of the proof, we will assume that $\overline r > r$.
To simplify the presentation, we let $s = \sqrt{8\gamma/\alpha}$.
Now assume for the sake of a contradiction that the statement of the lemma does not hold i.e.\ that $q(\overline{r}) < (1-s)q(r)$.  
We show that this implies that $\overline{\CR}(r) \geq \overline{\CR}(\overline{r})$ which contradicts the choice of $\overline{r}$.

By Lemma~\ref{empquantrange}, $\overline{\CR}(\overline{r}) = \overline{q}(\overline{r})\overline{r} \leq q(\overline{r})(1+\gamma)^2\overline{r} = \CR(\overline{r}) (1+\gamma)^2$.   
As $\overline{r} > r$, and as by assumption $q(\overline{r}) < (1-s)q(r)$,
%
%
\begin{align*}
\CR(q(\overline{r}))(1+\gamma)^2 &< \CR\left( (1-s) q(r) \right) (1+\gamma)^2 \\
&\leq \CR(q(r))\frac{1}{1-\alpha}\left((1-s)^{\alpha}-\alpha(1-s)\right)(1+\gamma)^2 \numberthis \label{qbound}\\
&\hspace*{2in} \mbox{(by Lemma~\ref{reserveb}).}
\end{align*}  
The Taylor expansion of $(1-s)^{\alpha}$ is $\sum_{i=0}^{\infty} \frac{(-1)^is^{i}\alpha(\alpha-1)\cdots(\alpha-i+1)}{i!}$. 
It follows that $(1-s)^{\alpha} \leq 1-\alpha s - (1-\alpha)\alpha s^2/2$, as $s < 1$.
Therefore~\eqref{qbound} is bounded above by 
\[\CR(q(r))\frac{1}{1-\alpha}\left(1-\alpha- (1-\alpha)\frac{\alpha s^2}{2}\right)(1+\gamma)^2 = \CR(q(r))\left(1-\frac{\alpha s^2}{2} \right) (1+\gamma)^2\numberthis \label{qboundtwo}.\]  
On replacing $s$ with $\sqrt{8\gamma/\alpha}$ in~\eqref{qboundtwo} this becomes \[\CR(q(r))(1-4\gamma)(1+\gamma)^2 \leq \frac{1}{(1+\gamma)^2}\CR(q(r)),~~ \mbox{as $1-4\gamma \leq 1/(1+\gamma)^4$.}\]


By Lemma~\ref{empquantrange} again, this is bounded above by $\overline{\CR}(r)$, yielding $\overline{\CR}(\overline{r}) < \overline{\CR}(r)$, a contradiction.
\end{proof}

Now we relate $\mathcal{R}^{\mathcal{M} \cap \mathcal{B}}$ and $\mathcal{R}^{\mathcal{M^\mathcal{L}}}$, following the proof of Lemma 13 in~\cite{CMM} very closely.

\begin{lem}\label{part2}
With probability $1-k\delta$, where $k$ is the number of bidders, \[\mathcal{R}^{\mathcal{M} \cap \mathcal{B}} \leq \frac{3}{\alpha^{1/(1 - \alpha)}(1-\max\{\sqrt{8\gamma/\alpha}, 4\gamma+\xi\gamma\})}\mathcal{R}^{\mathcal{M^\mathcal{L}}}.\]
\end{lem}

\begin{proof}
We assume that $\mathcal{E}_{\alpha}$ holds for the samples from each distribution, which it does with probability $1-k\delta$ by a union bound.  As in~\cite{CMM}, fix a bidder $i$, along with $(\mathbf{v_{-i}}, \mathbf{B})$ and consequently $T_i$.  Additionally, bidder $i$ only contributes to $\mathcal{R}^{\mathcal{M} \cap \mathcal{B}}$ if $v_i \geq T_i$. We consider two cases, $T_i \geq \overline{r}$ and $T_i \leq \overline{r}$.  The former case retains its proof in~\cite{CMM}.  In particular, \[\mathbb{E}_{v_i} [\mathcal{R}_i^{\mathcal{M} \cap \mathcal{B}}(\mathbf{v}, \mathbf{B}) | v_i \geq T_i] \leq \mathbb{E}_{v_i} [\mathcal{R}_i^{\mathcal{M^\mathcal{L}}}(\mathbf{v}, \mathbf{B}) | v_i \geq T_i] \numberthis \label{three},\]  where $\mathcal{R}_i^{\mathcal{X}}$ is the revenue from bidder $i$ in mechanism $\mathcal{X}$.

In the case $T_i \leq \overline{r}$, as in Lemma 13 in~\cite{CMM}, \[\mathbb{E}_{v_i} [\mathcal{R}_i^{\mathcal{M} \cap \mathcal{B}}(\mathbf{v}, \mathbf{B}) | v_i \geq T_i] \leq \min\{r, B_i\}.\numberthis \label{one}\]  By Lemma 10 in~\cite{CMM}, if $v_i \geq \overline{r}$, then $\mathcal{R}_i^{\mathcal{M^\mathcal{L}}} \geq \min\{\overline{r}, B_i\}/3$.  This occurs with probability $1-F(\overline{r})$, and therefore \[\mathbb{E}_{v_i} [\mathcal{R}_i^{\mathcal{M^\mathcal{L}}}(\mathbf{v}, \mathbf{B}) | v_i \geq T_i] \geq \min\{\overline{r}, B_i\}(1-F(\overline{r}))/3.\]

In the case that $\min\{\overline{r}, B_i\} = \overline{r}$, by Lemma~\ref{empreserve}, the right-hand side is bounded below by $(1-4\gamma-\xi)(\min\{r, B_i\})(1-F(r))/3$.  In the case that $\min\{\overline{r}, B_i\} = B_i$, by Lemma~\ref{reservequant}, the right-hand side is bounded below by $B_i\left(1-\sqrt{\frac{8\gamma}{\alpha}}\right)(1-F(r))/3$.  Therefore, in general  \begin{align*}\mathbb{E}_{v_i} [\mathcal{R}_i^{\mathcal{M^\mathcal{L}}}(\mathbf{v}, \mathbf{B}) | & v_i \geq T_i] \\ 
&\geq \min\left\{\left(1-\sqrt{\frac{8\gamma}{\alpha}}\right), (1-4\gamma-\xi)\right\}(\min\{r, B_i\})(1-F(r))/3 \\
&= (1-\max\{\sqrt{8\gamma/\alpha}, 4\gamma+\xi\})(\min\{r, B_i\})(1-F(r))/3.\end{align*} By Lemma~\ref{reservebound}, \[\mathbb{E}_{v_i} [\mathcal{R}_i^{\mathcal{M^\mathcal{L}}}(\mathbf{v}, \mathbf{B}) | v_i \geq T_i] \geq (1-\max\{\sqrt{8\gamma/\alpha}, 4\gamma+\xi\})(\min\{r, B_i\}) \alpha^{1/(1-\alpha)}/3. \numberthis \label{two}\] Combining~\eqref{one} and~\eqref{two} gives \begin{multline*}\mathbb{E}_{v_i} [\mathcal{R}_i^{\mathcal{M} \cap \mathcal{B}}(\mathbf{v}, \mathbf{B}) | v_i \geq T_i] \leq \\ 3\alpha^{-1/(1-\alpha)}(1-\max\{\sqrt{8\gamma/\alpha}, 4\gamma+\xi\})^{-1} \mathbb{E}_{v_i} [\mathcal{R}_i^{\mathcal{M^\mathcal{L}}}(\mathbf{v}, \mathbf{B}) | v_i \geq T_i].\end{multline*}

We have just shown the above inequality  in the case $T_i \leq \overline{r}$; but by~\eqref{three} it also holds when $T_i \geq \overline{r}$.  Next, we note that when $v_i \leq T_i$, bidder $i$ does not contribute to $\mathcal{R}_i^{\mathcal{M} \cap \mathcal{B}}$, and consequently the conditioning can be ignored.  Taking the expectation over $(\mathbf{v_{-i}}, \mathbf{B})$ and summing over all $i$ yields the lemma.
\end{proof}

Combining Lemmas~\ref{part1} and~\ref{part2} yields the following theorem.
\begin{thm}
\label{thm:lottery-samp}
The mechanism $\mathcal{M}^{\mathcal{L}}$ performed with $k$ bidders obtains a total revenue that is a $\frac{3} {1 - k \delta} \left( 1 + \frac {1} {\alpha^{1/(1 - \alpha)}(1-\max\{\sqrt{8\gamma/\alpha}, 4\gamma+\xi\gamma\})} \right)$-approximation to the revenue of the optimal mechanism, given $m \geq \frac{6(1+\gamma)}{\gamma^2\xi} \max\left\{\frac{\ln 3}{\gamma}, \ln{\frac{3}{\delta}}\right\}$ samples per class of bidders when there are $k$ classes, $\gamma \xi m \geq 4$ and $(1+\gamma)^2 \leq \frac{3}{2}$.
\end{thm}

\subsection{Multi-item Auctions with Public Budgets}

We now consider the mechanism in Theorem~\ref{lpmech}, first given in~\cite{BGGM}. Here there is a set of bidders, a set of items, and a different distribution for the valuations for each (bidder, item) pair.  Each item can be assigned to at most one bidder, each bidder can receive only a specified number of items, and each bidder can spend no more than a pre-determined budget.

\hide{Because the valuations are chosen from a discrete set, Lemma~\ref{valuequant} and Corollary~\ref{empquantc} need not hold, and we rely on Lemma~\ref{empquant} instead.}

In addition, we add a point mass to the empirical distribution at $\overline{q}(\overline{\xi})$ so that $\overline{v}(q)$ and $\overline{q}(v)$ are defined for all quantiles less than $\overline{\xi}$ and all values greater than $\overline{v}(\overline{\xi})$ respectively.

A key step in this mechanism is to solve the linear program LP2 below with coefficients derived from the distribution.  The optimal solution to LP2 yields a mechanism with the approximation guarantee stated in Theorem~\ref{lpmech}.  It is also the case that a feasible solution of LP2 achieving a solution with value $\rho \cdot OPT$ when substituted into a slight modification of the mechanism will yield an approximation factor of $\rho$ times the factor in Theorem~\ref{lpmech}. (See Figure $1$; the original mechanism had a probability of $1/4$ in line $4$; we use a smaller probability, $p/4$, which reduces the expected revenue to $p$ times its previous value.)  As we only have access to samples of the distribution, we can only create an approximate form of LP2, which we name LP3, described below.  The main challenge is to show that a modification of an optimal solution to LP3 is feasible for LP2 and achieves a good approximation to the optimum of LP2. 

Notation: $\overline{f}$ and $\overline{\phi}$ are the approximations to $f$ and $\phi$ respectively, derived from the sample-based empirical distribution.  Additionally, $I$ is the set of bidders, $J$ is the set of items, $B_i$ is the budget of bider $i$, $n_i$ is the number of items bidder $i$ can obtain, and $R_{ij}$ is the support of the distribution $F_{ij}$ of the valuation of item $j$ by bidder $i$.  The virtual valuation function $\phi_{ij}$ is derived from $F_{ij}$.  Finally, $x_{ij}(r)$ is the variable in the linear program corresponding to $i \in I, j \in J$ and $r \in R_{ij}$.
\begin{alignat*}{3}
\text{Maximize } & \sum_{i \in I} \sum_{j \in J} \sum_{r \in R_{ij}} f_{ij}(r )\phi_{ij}(r) x_{ij}(r)\ \\
\text{Subject to: } ~&1. ~~ \sum_{j \in J} \sum_{r \in R_{ij}} f_{ij}(r)x_{ij}(r) \leq n_i & \quad & \forall i \in I \\
& 2. ~~ \sum_{j \in J} \sum_{r \in R_{ij}} f_{ij}(r)\phi_{ij}(r)x_{ij}(r) \leq B_i & \quad & \forall i \in I \tag{LP2}\\
& 3. ~~ \sum_{i \in I} \sum_{r \in R_{ij}} f_{ij}(r)x_{ij}(r) \leq 1 &\quad & \forall j \in J \\
& 4. ~~ 0 \leq x_{ij}(r) \leq 1 & \quad & \forall i \in I, j \in J, r \in R_{ij} \\
\end{alignat*}
\begin{alignat*}{2}
\text{Maximize } & \sum_{i \in I} \sum_{j \in J} \sum_{r \in R_{ij}} \overline{f}_{ij}(r)\overline{\phi}_{ij}(r) \overline{x}_{ij}(r)\ \\
\text{Subject to } ~&1. ~~ \sum_{j \in J} \sum_{r \in R_{ij}} \overline{f}_{ij}(r)\overline{x}_{ij}(r) \leq n_i & \quad & \forall i \in I \\
& 2. ~~ \sum_{j \in J} \sum_{r \in R_{ij}} \overline{f}_{ij}(r)\overline{\phi}_{ij}(r)\overline{x}_{ij}(r) \leq B_i & \quad & \forall i \in I\tag{LP3} \\
& 3. ~~  \sum_{i \in I} \sum_{r \in R_{ij}} \overline{f}_{ij}(r)\overline{x}_{ij}(r) \leq 1 &\quad & \forall j \in J \\
& 4. ~~  0 \leq  \overline{x}_{ij}(r) \leq 1 & \quad & \forall i \in I, j \in J, r \in R_{ij}
\end{alignat*}
LP3 is obtained by replacing $f$ by $\overline{f}$, $\phi$ by $\overline{\phi}$, and $x$ by $\overline{x}$.

We will show that the optimal solution to LP3 is a good approximation to LP2.  We start by rewriting the above linear programs in a simpler format.  Let $x_{ij}^* = \sum_{r \in R_{ij}} f_{ij}(r ) x_{ij}(r)$; by Lemma 3.13 in~\cite{BGGM}, the optimum of LP2 is achieved when for some $L'$, $x_{ij}(r) = 0$ for $r < L'$, $x_{ij}(r) = 1$ for $r > L'$, and $0 \leq x_{ij}(L') \leq 1$ where $1 \leq L' \leq L$. Also, $\sum_{r \in R_{ij}} f_{ij}(r )\phi_{ij}(r) x_{ij}(r) = \CR_{ij}(x_{ij}^*)$.  We can define $\overline{x}_{ij}^*$ similarly; again $\sum_{r \in R_{ij}} \overline{f}_{ij}(r )\overline{\phi}_{ij}(r) \overline{x}_{ij}(r) = \overline{\CR}_{ij}(\overline{x}_{ij}^*)$.  These identities can be used to rewrite LP2 in terms of $x_{ij}^*$ rather than $x_{ij}(r)$ (shown below as LP2$'$), and LP3 in terms of $\overline{x}_{ij}^*$ (shown below as LP3$'$).
\begin{alignat*}{2}
\text{Maximize } & \sum_{i \in I} \sum_{j \in J} \CR_{ij}(x_{ij}^*)\ \\
\text{Subject to } ~&1. ~~ \sum_{j \in J} x_{ij}^* \leq n_i & \quad & \forall i \in I \\
& 2. ~~ \sum_{j \in J} \CR_{ij}(x_{ij}^*) \leq B_i & \quad & \forall i \in I \tag{LP2$'$} \\
& 3. ~~ \sum_{i \in I} x_{ij}^* \leq 1 &\quad & \forall j \in J \\
& 4. ~~ 0 \leq x_{ij}^* \leq 1 &\quad & \forall i \in I, j \in J
\end{alignat*}

\begin{alignat*}{2}
\text{Maximize } & \sum_{i \in I} \sum_{j \in J} \overline{\CR}_{ij}(\overline{x}_{ij}^*)\ \\
\text{Subject to } ~&1. ~~ \sum_{j \in J} \overline{x}_{ij}^* \leq n_i & \quad & \forall i \in I \\
& 2. ~~ \sum_{j \in J} \overline{\CR}_{ij}(\overline{x}_{ij}^*) \leq B_i & \quad & \forall i \in I \tag{LP3$'$} \\
& 3. ~~ \sum_{i \in I} \overline{x}_{ij}^* \leq 1 &\quad & \forall j \in J \\
& 4. ~~ 0 \leq \overline{x}_{ij}^* \leq 1 &\quad & \forall i \in I, j \in J
\end{alignat*}

Finally, we note that in the optimal solution of both LP2 and LP3, $x_{ij}^*$ and $\overline{x}_{ij}^*$ are no larger than the reserve price quantiles of the distributions $F_{ij}$ and $\overline{F}_{ij}$ respectively, for all $i$ and $j$.  To see this, consider any $x_{ij}^*$ for which this is not the case.  Because $\CR_{ij}$ is convex and $\CR(0) = \CR(1) = 0$, there is another point $y$ such that $\CR(x_{ij}^*) = \CR(y)$, but with $y < q(r)$.  Therefore, we can replace $x_{ij}^*$ with $y$, without changing the value of the objective function and continuing to satisfy all the constraints.  This is the case for both $x_{ij}^*$ in LP2 and $\overline{x}_{ij}^*$ in LP3.

As described in~\cite{BGGM}, the distributions $F_{ij}$ are truncated so that the support does not include any value greater than $B_{i}$.  Therefore $\CR_{ij}(q) \leq B_i q$ for all $i$ and $j$, and quantiles $q$.  This bound also holds for $\overline{\CR}_{ij}(q)$.  To see this, let the point $(q, \overline{\CR}_{ij}(q))$ lie on the line segment from $(\frac{2k_1-1}{m}, \frac{2k_1-1}{m}v_{k_1})$ to $(\frac{2k_2-1}{m}, \frac{2k_2-1}{m}v_{k_2})$.  If $v_{k_1} = v_{k_2} = B_i$ then $\overline{\CR}_{ij}(q) = qB_i$, but if $v_{k_1}$ or $v_{k_2}$ were any smaller, the value of $\overline{\CR}_{ij}(q)$ would also be smaller.

We modify the mechanism given in~\cite{BGGM}.  In particular, we replace the step that solves LP2 with a step to solve LP3 to obtain the solution $\overline{x}_{ij}^*$.  We then construct $\overline{z}_{ij}^* = \max\{\overline{x}_{ij}^*, \overline{\xi}(1+\gamma)^2\}$, and use $\overline{z}_{ij}^*$ in the mechanism.  Additionally, instead of offering each item to each bidder with probability $1/4$, we do so with probability $(1-c\overline{\xi})/(4(1+\gamma)^2)$ where $c = \max\{|I|, |J|\}(1+\gamma)^4$.


Recall that, as proved in Lemma 3.13 in~\cite{BGGM}, $\overline{v}(\overline{z}_{ij}^*)$ can be written as $\overline{w}_{ij}^*\overline{r}_{ij}+(1-\overline{w}_{ij}^*)(\overline{r}_{ij}+1)$ where $\overline{r}_{ij} \in R_{ij}$ and $0 \leq \overline{w} \leq 1$.  Additionally, $\overline{\CR}(\overline{z}_{ij}^*)=\overline{w}_{ij}^*\overline{\CR}(r_{ij})+(1-\overline{w}_{ij}^*)\overline{\CR}(r_{ij}+1)$.  The $\overline{r}_{ij}$ and $\overline{w}_{ij}^*$ will be used in the mechanism given below in Figure 1.

\begin{figure}{\parbox{\textwidth}{
\begin{enumerate}
\item Solve LP3, and let $\overline{z}_{ij}^*$ be as defined in the text.

\item Process the bidders in term in some fixed, but arbitrary order.

\item For each $(i, j)$ let $\widetilde{r}_{ij}$ be $\overline{r}_{ij}$ with probability $\overline{w}_{ij}^*$ and $\overline{r}_{ij}+1$ with probability $1-\overline{w}_{ij}^*$, where $\overline{r}_{ij}$ and $\overline{w}_{ij}^*$ are as defined in the text.

\item For each $(i, j)$ let $Y_{ij}$ be a $0/1$ random variable so that $Y_{ij} = 1$ with probability $(1-c\overline{\xi})/(4(1+\gamma)^2)$ and $Y_{ij} = 0$ otherwise.

\item Offer each bidder $i$ all items $j$ with $Y_{ij} = 1$ that have not already been bought, at price $\widetilde{r}_{ij}$.
\end{enumerate}
}}
\caption{Posted-Price Mechanism}\end{figure}

\hide{In particular, we do this for a slight modification to the optimal solution $\overline{x}_{ij}^*$ of LP3.  Let $\overline{z}_{ij}^* = \max\{\overline{x}_{ij}^*, \overline{\xi}(1+\gamma)^4\}/(1+\gamma)^2$.  By Lemma 3.13 in~\cite{BGGM}, as $\overline{z}_{ij}^*$ is a quantile, we can write it as $\overline{z}_{ij}^* = \sum_{r \in R_{ij}} f_{ij}(r ) \overline{z}_{ij}(r)$.  Note that the original distribution is used, and not the empirical distribution.  Then we let the modification be $\overline{y}_{ij}(r) = \frac{1-c\overline{\xi}}{(1+\gamma)^2}\overline{z}_{ij}(r)$ where $c = \max\{|I|, |J|\}(1+\gamma)^4$.}

Consider the quantile $q(\overline{v}(\overline{z}_{ij}^*))$.  This is the ex-ante probability that bidder $i$'s valuation of item $j$ is at least $\widetilde{r}_{ij}$. By Lemma 3.13 in~\cite{BGGM}, $\widetilde{r}_{ij}$ can be written as $\sum_{r \in R_{ij}} f_{ij}(r ) q(\overline{v}(\overline{z}_{ij}^*(r)))$.  Note that the original distribution is used, and not the empirical distribution.  Then consider the variables $\overline{y}_{ij}(r)$ defined as $\overline{y}_{ij}(r) = \frac{1-c\overline{\xi}}{(1+\gamma)^2}q(\overline{v}(\overline{z}_{ij}^*(r)))$ where $c = \max\{|I|, |J|\}(1+\gamma)^4$.  The following lemma shows that $\overline{y}_{ij}$ satisfies LP2.  This will be used to show that the mechanism in Figure 1 satisfies the allocation bounds with constant probability.


\begin{lem}\label{lp2constraint}
The $\overline{y}_{ij}$ described above satisfy the constraints of LP2 with probability $1-|I|J|\delta$.  
\end{lem}
\begin{proof}
Assume that $\mathcal{E}_{\alpha}$ holds for the samples from each distribution, which it does with probability $1-|I||J|\delta$.  Consider the second constraint.  By Lemma~\ref{empquantrange}, as $\overline{z}_{ij}^* \geq \overline{\xi}(1+\gamma)^2$, \[\CR_{ij}(q(\overline{v}(\overline{z}_{ij}^*))) = q(\overline{v}(\overline{z}_{ij}^*))\overline{v}(\overline{z}_{ij}^*) \leq (1+\gamma)^2\overline{z}_{ij}^*\overline{v}(\overline{z}_{ij}^*) = (1+\gamma)^2 \overline{\CR}_{ij}(\overline{z}_{ij}^*)\numberthis \label{crbound}.\]  Additionally, if $\overline{z}_{ij}^* = \overline{\xi}(1+\gamma)^2$, then $\CR_{ij}(\overline{z}_{ij}^*) \leq B_{i}\overline{\xi}(1+\gamma)^2$.  Otherwise, $\overline{z}_{ij}^* = \overline{x}_{ij}^*$,  and $\overline{\CR}_{ij}(\overline{z}_{ij}^*) = \overline{\CR}_{ij}(\overline{x}_{ij}^*)$.  Therefore for each $i$, \begin{align*}\sum_{j \in J} \sum_{r \in R_{ij}} f_{ij}(r)\phi_{ij}(r)\overline{y}_{ij}(r) & = \frac{1-c\overline{\xi}}{(1+\gamma)^2}\sum_{j \in J} \CR_{ij}(q(\overline{v}(\overline{z}_{ij}^*))) \\
&\leq ({1-c\overline{\xi}}) \sum_{j \in J} \overline{\CR}_{ij}(\overline{z}_{ij}^*) \tag{by~\eqref{crbound}}\\
&\leq ({1-c\overline{\xi}}) \sum_{j \in J} \left(\overline{\CR}_{ij}(\overline{x}_{ij}^*)+B_{i}\overline{\xi}{(1+\gamma)^2}\right)\\
&\leq (1-|J|(1+\gamma)^4\overline{\xi})B_i+|J|B_i(1+\gamma)^2\overline{\xi} \tag{as $c \geq |J| (1+\gamma)^4)$ and by Constraint $2$ in LP3$'$}\\
&\leq B_i \end{align*} and the second constraint is satisfied.

The proofs for the first and third constraints are similar.  The first constraint holds, as \begin{align*}\sum_{j \in J} \sum_{r \in R_{ij}} f_{ij}(r ) \overline{y}_{ij}(r) &= \frac{1-c\overline{\xi}}{(1+\gamma)^2}\sum_{j \in J} q(\overline{v}(\overline{z}_{ij}^*))  \\
&\leq ({1-c\overline{\xi}})\sum_{j \in J} \overline{z}_{ij}^* \tag{by Lemma~\ref{empquantrange} as $\overline{z}_{ij}^* \geq \overline{\xi}$ is an empirical quantile}\\
&\leq |J|(1+\gamma)^2\overline{\xi}+(1-c\overline{\xi})\sum_{j \in J} \overline{x}_{ij}^* \\
&\leq |J|(1+\gamma)^2\overline{\xi}+(1-c\overline{\xi})n_i \tag{by Constraint 1 in LP3$'$}\\
&\leq |J|(1+\gamma)^2\overline{\xi}+n_i-|J|(1+\gamma)^4\overline{\xi} n_i \\
&\leq n_i.\end{align*}

The third constraint holds also, as \begin{align*}\sum_{i \in I} \sum_{r \in R_{ij}} f_{ij}(r ) \overline{y}_{ij}(r) &= \frac{1-c\overline{\xi}}{(1+\gamma)^2}\sum_{i \in I} q(\overline{v}(\overline{z}_{ij}^*)) \\
&\leq ({1-c\overline{\xi}})\sum_{i \in I} \overline{z}_{ij}^* \tag{by Lemma~\ref{empquantrange}} \\
&\leq |I|(1+\gamma)^2\overline{\xi}+(1-c\overline{\xi})\sum_{i \in I} \overline{x}_{ij}^* \\
&\leq |I|(1+\gamma)^2\overline{\xi}+(1-c\overline{\xi}) \tag{by Constraint 3 in LP3$'$}\\
&\leq |I|(1+\gamma)^2\overline{\xi}+1-|I|(1+\gamma)^4\overline{\xi} \\
&\leq 1.\end{align*}
\end{proof}

We also show that the value of the optimization function of LP2 on the solution $\overline{y}_{ij}$ is close to its value on the optimal solution $x_{ij}^*$.

\begin{lem}\label{lp3}
Let $\overline{y}_{ij}$ be as described above, let $x_{ij}^*$ be the optimal solution of LP2, and define \[V_2 = \sum_{i \in I} \sum_{j \in J} \CR_{ij}(x_{ij}^*) \text{, and } V_3 = \sum_{i \in I} \sum_{j \in J} \sum_{r \in R_{ij}} f_{ij}(r ) \overline{\phi}_{ij}(r)\overline{y}_{ij}(r).\]  Then with probability $1-|I||J|\delta$, $V_3 \geq \frac{(1-c\overline{\xi})(1-c'\overline{\xi})(1-\overline{\xi}(1+\gamma)^3)^2}{(1+\gamma)^9} V_2,$ where $c = \max\{|I|, |J|\}(1+\gamma)^4$, $c' = \max\{|I|, |J|\}(1+\gamma)^2$, and each empirical distribution is determined by $m \geq \frac{6(1+\gamma)}{\gamma^2\xi} \max\{\frac{\ln{3}}{\gamma}, \ln\frac{3}{\delta}\}$ samples, with $\gamma \xi m \geq 4$ and $(1+\gamma)^2 \leq \frac{3}{2}$.
\end{lem}

\begin{proof}
Assume that $\mathcal{E}_{\alpha}$ holds for the samples from each distribution, which it does with probability $1-|I||J|\delta$.  Consider the optimal assignment to LP2, $x_{ij}^*$.  We show that a slight modification to $x_{ij}^*$ satisfies the conditions of LP3, and will use this to bound $V_3$ from below.  Let $z^*_{ij} = \max\{x^*_{ij}, (1+\gamma)^2\overline{\xi}\}/(1+\gamma)^2$ if $z^*_{ij} \leq \overline{q}(\overline{r})$ and $\max\{x^*_{ij}, (1+\gamma)^3\overline{\xi}\}/(1+\gamma)^3$ otherwise.  For brevity, we write $z^*_{ij} = \max\{x^*_{ij}, (1+\gamma)^{e_{ij}}\overline{\xi}\}/(1+\gamma)^{e_{ij}},$ where $e_{ij} = 2$ or $3,$ as appropriate.  By Lemma 3.13 in~\cite{BGGM}, we can write $z_{ij}^*$ as $z_{ij}^* = \sum_{r \in R_{ij}} \overline{f}_{ij}(r ) z_{ij}(r)$, and we let $y_{ij}(r) = \frac{1-c'\overline{\xi}}{(1+\gamma)^2}z_{ij}(r)$ where $c' = \max\{|I|, |J|\}(1+\gamma)^2$.  We can show that $y_{ij}$ satisfies the constraints of LP3 with an argument similar to that of Lemma~\ref{lp2constraint}.

The first constraint holds, as \begin{align*}\sum_{j \in J} \sum_{r \in R_{ij}} \overline{f}_{ij}(r ) y_{ij}(r) &= \frac{1-c'\overline{\xi}}{(1+\gamma)^2}\sum_{j \in J} z_{ij}^* \\
&\leq \frac{1}{(1+\gamma)^4}\left(\sum_{j \in J}(1+\gamma)^4\overline{\xi}+(1-c'\overline{\xi})\sum_{j \in J} x_{ij}^*\right) \\
&\leq \frac{1}{(1+\gamma)^4}\left(|J|(1+\gamma)^4\overline{\xi}+(1-c'\overline{\xi})(1+\gamma)^2n_i\right) \tag{by Constraint 1 in LP2$'$} \\
&\leq n_i.\end{align*}  The third constraint holds also, as \begin{align*}\sum_{i \in I} \sum_{r \in R_{ij}} \overline{f}_{ij}(r ) y_{ij}(r) &= \frac{1-c'\overline{\xi}}{(1+\gamma)^2}\sum_{i \in I} z_{ij}^* \\
&\leq \frac{1}{(1+\gamma)^4}\left(\sum_{i \in I}(1+\gamma)^4\overline{\xi}+(1-c'\overline{\xi})\sum_{i \in I} x_{ij}^*\right) \\
&\leq \frac{1}{(1+\gamma)^4}\left(|I|(1+\gamma)^4\overline{\xi}+(1-c'\overline{\xi})(1+\gamma)^2\right)  \tag{by Constraint 3 in LP2$'$}\\
&\leq \frac{1}{(1+\gamma)^4}\left(|J|(1+\gamma)^4\overline{\xi}+(1-|J|(1+\gamma)^2\overline{\xi})(1+\gamma)^2n_i\right) \\
&\leq \frac{1}{(1+\gamma)^4}\left(|I|(1+\gamma)^4\overline{\xi}+(1-|I|(1+\gamma)^2\overline{\xi})(1+\gamma)^2\right) \\
&\leq 1.\end{align*} 

In order to show that the second constraint holds, we first need to show that $\overline{\CR}_{ij}(z_{ij}^*)$ is bounded above by $\CR_{ij}(x_{ij}^*)(1+\gamma)^2+B_i\overline{\xi}$.  In the case that $z_{ij}^* = x^*_{ij}/(1+\gamma)^2$, by Lemma~\ref{empquant}, $\overline{\CR}_{ij}(z_{ij}^*)$ is bounded above by $(1+\gamma)^2\CR_{ij}(x_{ij}^*/(1+\gamma)^4)$.  However, because $q(r_{ij}) \geq x_{ij}^* \geq x_{ij}^*/(1+\gamma)^4$, $\CR_{ij}(x_{ij}^*/(1+\gamma)^4) \leq \CR_{ij}(x_{ij}^*)$.  If $z_{ij}^* = x^*_{ij}/(1+\gamma)^3$, then by Lemma~\ref{empquant}, $\overline{\CR}_{ij}(z_{ij}^*)$ is bounded above by $(1+\gamma)^2\CR_{ij}(x_{ij}^*/(1+\gamma)^5)$ which is also bounded above by $(1+\gamma)^2\CR_{ij}(x_{ij}^*)$.  Otherwise, if $z_{ij}^* = \overline{\xi}$ then $\overline{\CR}_{ij}(z_{ij}^*)$ is bounded above by $B_i\overline{\xi}$. Therefore, the second constraint holds, as for every $i,$
\begin{align*}
\sum_{j \in J} \sum_{r \in R_{ij}} \overline{f}_{ij}(r )\overline{\phi}_{ij}(r) y_{ij}(r) &= \frac{1-c'\overline{\xi}}{(1+\gamma)^2} \sum_{j \in J} \overline{\CR}_{ij}(z_{ij}^*) \\
&\leq \frac{1}{(1+\gamma)^2}\left(|J|B_i\overline{\xi} + (1-c'\overline{\xi})\sum_{j \in J} \CR_{ij}(x_{ij}^*)(1+\gamma)^2\right)\\
&\leq \frac{1}{(1+\gamma)^2}\left(|J|B_i\overline{\xi} + (1-|J|(1+\gamma)^2\overline{\xi})B_i(1+\gamma)^2\right) \tag{by Constraint $2$ in LP2$'$}\\
&\leq B_i.
\end{align*}

We now bound $V_3$ from below.  Using the identity $\sum_{r \in R_{ij}} f_{ij}(r)\phi_{ij}(r)\overline{y}_{ij}(r) = \frac{1-c\overline{\xi}}{(1+\gamma)^2} \CR_{ij}(q(\overline{v}(\overline{z}_{ij}^*)))$ and applying Lemma~\ref{empquantrange} to $q(\overline{v}(\overline{z}_{ij}^*))$, it follows that\begin{align*}
V_3 =  \sum_{i \in I} \sum_{j \in J} \sum_{r \in R_{ij}} f_{ij}(r )\phi_{ij}(r) \overline{y}_{ij}(r)
&= \frac{1-c\overline{\xi}}{(1+\gamma)^2}\sum_{i \in I} \sum_{j \in J} \CR_{ij}(q(\overline{v}(\overline{z}_{ij}^*))) \\
&\geq \frac{1-c\overline{\xi}}{(1+\gamma)^4}\sum_{i \in I} \sum_{j \in J} \overline{\CR}_{ij}(\overline{z}^*_{ij}) \numberthis \label{vthreebound}\end{align*}  We next compare $\overline{\CR}_{ij}(\overline{z}^*_{ij})$ to $\overline{\CR}_{ij}(\overline{x}_{ij}^*)$.  If $\overline{z}^*_{ij} = \overline{x}^*_{ij}$, then the two are equal.  Otherwise, $\overline{z}^*_{ij} = \overline{\xi}(1+\gamma)^2$ and it follows from the convexity of $\overline{\CR}$ that $(1-\overline{\xi}(1+\gamma)^2)\overline{\CR}_{ij}(\overline{x}^*_{ij}) \leq \overline{\CR}_{ij}(\overline{z}_{ij}^*)$.  Therefore, we can lower bound~\eqref{vthreebound} and hence $V_3$ by 
\[V_3 \geq \frac{(1-c\overline{\xi})(1-\overline{\xi}(1+\gamma)^2)}{(1+\gamma)^4}\sum_{i \in I} \sum_{j \in J} \overline{\CR}_{ij}(\overline{x}^*_{ij}).\] Because $\overline{x}_{ij}^*$ is the optimal solution to LP3, it follows that \[\sum_{i \in I} \sum_{j \in J}\overline{\CR}_{ij}(\overline{x}_{ij}) \geq \sum_{i \in I} \sum_{j \in J}\sum_{r \in R_{ij}} \overline{f}_{ij}(r)\overline{\phi}_{ij}(r) y_{ij}(r),\] yielding a lower bound on $V_3$ of
\[V_3 \geq \frac{(1-c\overline{\xi})(1-\overline{\xi}(1+\gamma)^{2})}{(1+\gamma)^{4}}\sum_{i \in I} \sum_{j \in J} \sum_{r \in R_{ij}} \overline{f}_{ij}(r)\overline{\phi}_{ij}(r) y_{ij}(r).\]  Using the identity $\sum_{r \in R_{ij}} \overline{f}_{ij}(r )\overline{\phi}_{ij}(r) y_{ij}(r) = \frac{1-c'\overline{\xi}}{(1+\gamma)^2}\overline{\CR}_{ij}(z_{ij}^*)$, yields \[V_3 \geq \frac{(1-c\overline{\xi})(1-c'\overline{\xi})(1-\overline{\xi}(1+\gamma)^{2})}{(1+\gamma)^{6}}\sum_{i \in I} \sum_{j \in J} \overline{\CR}_{ij}(z^*_{ij}).\]  We apply Lemma~\ref{empquant} to obtain a lower bound of \[V_3 \geq \frac{(1-c\overline{\xi})(1-c'\overline{\xi})(1-\overline{\xi}(1+\gamma)^{2})}{(1+\gamma)^{9}}\sum_{i \in I} \sum_{j \in J} \CR_{ij}(z_{ij}(r)(1+\gamma)^{e_{ij}}).\] Similarly as before, because $(1-\overline{\xi}(1+\gamma)^3)\CR(x_{ij}^*) \leq \CR(z^*_{ij}(1+\gamma)^{e_{ij}})$ we lower bound the above by
\begin{align*}V_3 &\geq \frac{(1-c\overline{\xi})(1-c'\overline{\xi})(1-\overline{\xi}(1+\gamma)^2)(1-\overline{\xi}(1+\gamma)^3)}{(1+\gamma)^{9}}\sum_{i \in I} \sum_{j \in J} \CR_{ij}(x^*_{ij}) \\
&= \frac{(1-c\overline{\xi})(1-c'\overline{\xi})(1-\overline{\xi}(1+\gamma)^2)(1-\overline{\xi}(1+\gamma)^3)}{(1+\gamma)^{9}} V_2.
\end{align*}
\end{proof}


\hide{Because this is the only step of the mechanism described in~\cite{BGGM} that relies on knowledge of the distribution, Lemma~\ref{lp3} is enough to show that the empirical mechanism achieves a good approximation.  In particular, we replace the step that solves LP2 with a step to solve LP3, and divide all the variables by a $(1+\gamma)^2$ factor.}

\begin{thm}
\label{thm-BGGm-w-sampling}
The empirical posted-price mechanism in Figure 1 gives an approximation to the optimal truthful-in-expectation Bayesian incentive-compatible mechanism with a multiplicative error of \[\frac{192}{\alpha}\left(\frac{2-\alpha}{\alpha}\right)^{1/(1-\alpha)}\frac{(1-c\overline{\xi})(1-c'\overline{\xi})(1-\overline{\xi}(1+\gamma)^{3})^2}{(1+\gamma)^{9}}(1-|I||J|\delta),\] given $m \geq \frac{6(1+\gamma)}{\gamma^2\xi} \max\{\frac{\ln{3}}{\gamma}, \ln\frac{3}{\delta}\}$ samples from each distribution, if $\gamma \xi m \geq 4$ and $(1+\gamma)^2 \leq \frac{3}{2}$.
\end{thm}
\begin{proof}
Using the same proof as in Theorem 3.13 in~\cite{BGGM}, we can show that the mechanism in Figure 1 allocates item $j$ to bidder $i$ with probability $\frac{1}{6} \cdot \frac{\overline{y}_{ij}^*}{4}.$  This follows from the fact that $\overline{y}_{ij}$ satisfies the constraints of LP2, as stated in Lemma~\ref{lp2constraint}.  By Lemma~\ref{lp3}, the objective function of LP2 at $\overline{y}_{ij}$ is close to the optimum $V_2$, with probability $1-|I||J|\delta.$

Finally, $\frac{1}{24}V_2$, the revenue guaranteed by Theorem~\ref{lpmech}, is a $\frac{192} {\alpha} \left( \frac {2 - \alpha} {\alpha} \right)^{1/(1 - \alpha)}$ approximation to the revenue of the optimal mechanism.  By the same analysis, the mechanism in Figure 1 yields revenue $\frac{1}{24}V_3$, yielding the desired bound.
\end{proof}

\section{Acknowledgements}
We thank the referees for their thoughtful comments.

\bibliographystyle{plain}
\bibliography{AlphaDistributions,auction-refs}


\end{document}